\documentclass[amsmath,amssymb,twocolumn,notitlepage,footinbib,superscriptaddress,floatfix,eprint,aps,physrev,10pt,tightenlines]{revtex4-2}

\usepackage[utf8]{inputenc}
\usepackage[english]{babel}
\usepackage{graphicx}
\usepackage{color}
\usepackage[usenames,dvipsnames,table]{xcolor}
\usepackage[normalem]{ulem}
\usepackage{amsthm}

\usepackage{microtype}
\microtypecontext{spacing=nonfrench}
\microtypesetup{
protrusion={true,nocompatibility},
activate={true,nocompatibility},
tracking=true,
kerning=true,
spacing={true}
}
\usepackage[all=normal,floats=tight,mathspacing=tight,wordspacing=tight,paragraphs=normal,tracking=tight,charwidths=tight,mathdisplays=tight,sections=normal,margins=normal]{savetrees}

\usepackage[style=american,autopunct=true]{csquotes}

\usepackage{placeins} 

\definecolor{blue}{rgb}{0,0,1}
\definecolor{grey}{rgb}{0.6,0.6,0.6}

\definecolor{myurlcolor}{rgb}{0,0,0.7}
\definecolor{myrefcolor}{rgb}{0.8,0,0}
\definecolor{purple}{RGB}{128,0,128}
\definecolor{ultramarine}{RGB}{63, 0, 255}
\definecolor{medblue}{RGB}{0, 0, 100}
\definecolor{googleblue}{RGB}{34, 0, 204}
\definecolor{panblue}{RGB}{0,24,150}
\definecolor{carmine}{RGB}{150, 0, 24}
\definecolor{gray}{RGB}{150, 150, 150}

\usepackage[unicode=true,pdfusetitle, bookmarks=false,bookmarksnumbered=false,
bookmarksopen=false, breaklinks=false,pdfborder={0 0 0},backref=false,
colorlinks=true, 
linkcolor=carmine,
citecolor=googleblue,
urlcolor=panblue,
anchorcolor=OliveGreen,
]{hyperref}

\newcommand{\ket}[1]{\left|#1\right\rangle}
\newcommand{\bra}[1]{\langle#1|}
\newcommand{\CHSH}{\mathsf{CHSH}}

\usepackage{mathtools}
\newtheorem{thm}{Theorem}

\newtheorem{prop}[thm]{Proposition}

\newtheorem{lemma}[thm]{Lemma}
\newtheorem{cor}{Corollary}[thm]

\begin{document}

\title{Causal Networks and Freedom of Choice in Bell's Theorem}

\author{Rafael Chaves}
\email{rchaves@iip.ufrn.br}
\affiliation{International Institute of Physics, Federal University of Rio Grande do Norte, 59078-970, P. O. Box 1613, Natal, Brazil}
\affiliation{School of Science and Technology, Federal University of Rio Grande do Norte, Natal, Brazil}
\author{George Moreno}
\affiliation{International Institute of Physics, Federal University of Rio Grande do Norte, 59078-970, P. O. Box 1613, Natal, Brazil}
\author{Emanuele Polino}
\author{Davide Poderini}
\author{Iris Agresti}
\author{Alessia Suprano}
\affiliation{Dipartimento di Fisica - Sapienza Universit\`{a} di Roma, P.le Aldo Moro 5, I-00185 Roma, Italy}
\author{Mariana R. Barros}
\affiliation{Dipartimento di Fisica “Ettore Pancini”, Università Federico II, Complesso Universitario di Monte Sant’Angelo, Via Cintia, 80126 Napoli, Italy}
\author{Gonzalo Carvacho}
\affiliation{Dipartimento di Fisica - Sapienza Universit\`{a} di Roma, P.le Aldo Moro 5, I-00185 Roma, Italy}
\author{Elie Wolfe}
\affiliation{Perimeter Institute for Theoretical Physics, 31 Caroline St. N, Waterloo, Ontario, N2L 2Y5, Canada}
\author{Askery Canabarro}
\affiliation{International Institute of Physics, Federal University of Rio Grande do Norte, 59078-970, P. O. Box 1613, Natal, Brazil}
\affiliation{Grupo de F\'isica da Mat\'eria Condensada, N\'ucleo de Ci\^encias Exatas - NCEx, Campus Arapiraca, Universidade Federal de Alagoas, 57309-005, Arapiraca, AL, Brazil}

\author{Robert W. Spekkens}
\affiliation{Perimeter Institute for Theoretical Physics, 31 Caroline St. N, Waterloo, Ontario, N2L 2Y5, Canada}

\author{Fabio Sciarrino } 
\affiliation{Dipartimento di Fisica - Sapienza Universit\`{a} di Roma, P.le Aldo Moro 5, I-00185 Roma, Italy}

\begin{abstract}
Bell's theorem is typically understood as the proof that quantum theory is incompatible with local hidden variable models. More generally, we can see the violation of a Bell inequality as witnessing the impossibility of explaining quantum correlations with classical causal models. The violation of a Bell inequality, however, does not exclude classical models where some level of measurement dependence is allowed, that is, the choice made by observers can be correlated with the source generating the systems to be measured. Here we show that the level of measurement dependence can be quantitatively upper bounded if we arrange the Bell test within a network.
Furthermore, we also prove that these results can be adapted in order to derive non-linear Bell inequalities for a large class of causal networks  and to identify quantumly realizable correlations which violate them.
\end{abstract}

\maketitle 

\section{Introduction}
Bell's theorem \cite{bell1964einstein} can arguably be seen as the most radical departure from classical physics. The kind of nonclassicality it entails is achieved without the need of any specific details or experimental assumptions that furthermore can be put to practical use in a variety of quantum information processing protocols in what is known as the device-independent framework \cite{pironio2016focus}.

In its standard interpretation, the violation of a Bell inequality shows that quantum correlations are incompatible with any theories respecting  local realism, that is, theories where physical properties have well defined values prior to any measurement and in such a way that far away events do not have a direct causal influence over each other. From another perspective, Bell inequality violation can also be seen as disproving any theories obeying the notion of local causality which implies a certain factorization of probabilities and that can be derived from two more fundamental assumptions, i.e., Reichenbach's principle \cite{reichenbach1991direction} and relativistic causality according to which the past corresponds to the past light-cone (see Refs.~\cite{wood2015lesson,wiseman2017causarum}). However, as firstly pointed out by Brans \cite{brans1988bell}, local hidden variable models are still capable of reproducing the quantum predictions if we allow for measurement dependence, a mechanism where our measurement devices are correlated with the system to be measured (see also~\cite{bell1985exchange,kofler2006experimenter,koh2012effects}). This subtle assumption in Bell's theorem, also known as the assumption of ``free will'' or as ``statistical independence'', has since then attracted growing attention, both from a theoretical \cite{hall2010local,PhysRevLett.106.100406,hall2011relaxed,koh2012effects,banik2012optimal,colbeck2012free,gallicchio2014testing,putz2014arbitrarily,chaves2015unifying,chaves2017causal,friedman2019relaxed,hall2020measurement,hossenfelder2020rethinking} and experimental \cite{aktas2015demonstration,handsteiner2017cosmic,big2018challenging,rauch2018cosmic} perspective and can be related to the communication cost between the measurement stations, needed by classical models to reproduce quantum correlations \cite{maudlin1992bell,brassard1999cost,steiner2000towards,toner2003communication,PhysRevLett.106.100406,maudlin2011quantum,ringbauer2017probing,brask2017bell,gill2020triangle,blasiak2021violations}.

Of particular relevance, is  the paradigmatic Clauser-Horne-Shimony-Holt (CHSH) Bell scenario \cite{clauser1969proposed} -- involving two distant parties, each measuring two dichotomic observables -- which has been thoroughly analyzed also allowing for relaxations of the measurement dependence assumption \cite{hall2020measurement}. As compared with the relaxation of locality, where it is known \cite{toner2003communication} that one requires 1 bit of classical communication to simulate the maximal quantum violation of the CHSH inequality, measurement dependence turns out to be a stronger resource, as merely $0.046$ bits of correlation are already enough to achieve the simulation~\cite{hall2020measurement}.
In spite of steady progress, all results to date suffer from the fact that they only impose lower bounds to the amount of measurement dependence needed to simulate a given violation of a Bell inequality using a classical causal model. 
If, by considering a slight modification of a Bell experiment, one had the means of determining an {\em upper bound} on the amount of measurement dependence that can be present, then whenever this amount was less than the lower bound on the amount needed to explain the observed violation of a Bell inequality in a classical causal model, one could infer that these violations were due to nonclassical effects.  Such a modification, in other words, would provide a means of adjudicating between measurement dependence and nonclassicality as a means of explaining the violation.
Measurement dependence has then remained as a seemingly untestable loophole in any Bell experiment.

The first aim of this paper is to revisit measurement dependence and show that under some assumptions it can indeed be upper bounded from the data observed in a slightly modified Bell experiment. For that, we arrange the standard Bell scenario as part of a larger causal network that includes an auxiliary variable,
and use the correlations between the measurement inputs and the auxiliary variable in order to upper bound how much such inputs might depend on the source generating the physical systems to be measured. From that we obtain non-linear Bell inequalities -- which explicitly incorporate possible correlations between the system to be measured and the settings of the measurement devices --  the violation of which is a clear signature of nonclassicality, rather than something which can be explained by merely positing measurement dependence within a classical causal model.

Following this, we explore the connections between measurement-dependent Bell causal structures and general causal networks that have recently started to attract attention in the literature \cite{tavakoli2021bell}. As compared with usual Bell scenarios, these new networks have two characteristic features. First, the fact that the correlations between the distant parties are now mediated by a number of independent sources. Second, the fact that one can prove nonclassical behaviour even without the need for any inputs, something considered quintessential in Bell's theorem \cite{wolf2009measurements}. In spite of the promising foundational and applied uses, progress in the analysis of general causal networks has been impeded by the fact that the correlations compatible with them define nonconvex sets, for which decades of expertise gathered with the standard Bell scenario and convex optimization algorithms are of limited or no use. Here, we show that standard Bell scenarios with measurement dependence can readily and generally be mapped onto causal networks of growing size and with different topologies. Not only do we derive new non-linear Bell inequalities for a variety of networks but also show, for the first time, that they lead to nonclassical correlations. 

The paper is organized as follows. In Sec.~\ref{sec:MDdefin}, we revisit the measurement independence assumption in Bell's theorem from a causal perspective and in particular identify and discriminate two possible mechanisms able to generate measurement dependence. In Sec.~\ref{sec:entropic}, we briefly present the entropic approach for the characterization of classical causal structures. In Sec.~\ref{sec:nonlinearineqs}, we discuss and derive a number of results and inequalities for the characterization of bipartite Bell scenarios with measurement dependence. In this section we also generalize these results to generic multipartite Bell scenarios and in particular derive strong lower bounds for measurement dependence based on the Mermin inequality~\cite{mermin1990}. In Sec.~\ref{sec:causnetworks}, we adapt our results for the analysis of a wide variety of causal networks, in particular proving that they can lead to nonclassical correlations. Finally, in Sec.~\ref{sec:conclusions}, we discuss our results and point out interesting directions for future research. 

\section{\!\!\!Freedom of choice in Bell's theorem}\label{sec:MDdefin}
Bell's theorem \cite{bell1964einstein} is the prime example of the incompatibility of quantum predictions with those  of classical causal models. From a modern perspective \cite{wood2015lesson,chaves2015unifying}, we impose a given causal structure to our quantum experiment and inquire whether causal models of classical origin can explain the observed empirical data, that is, the observed probability distribution. Different causal structures can be used to that aim \cite{branciard2010characterizing,branciard2012bilocal,fritz2012beyond,fritz2016beyond,tavakoli2014nonlocal,chaves2016polynomial,poderini2020experimental,rosset2016nonlinear,andreoli2017maximal,chaves2018quantum,renou2019genuine} but the paradigmatic illustration involves two distant parties, Alice and Bob, that upon receiving physical systems from a common source, randomly decide which measurements to perform obtaining the corresponding outcomes. Their measurement choices are labelled by the variables $X$ and $Y$ while their outcomes are labelled by $A$ and $B$, for Alice and Bob respectively. We note that random variables are typically denoted by uppercase letters and the values these variables can assume by lowercase letters. For instance, $X$ would represent the input variable of Alice and ${X{=}x}$, a specific value of it. In order to simplify the notation and presentation, in what follows we will use lowercase letters only.

Invoking special relativity, if Alice and Bob are space-like separated, then the measurement choices of one should have no causal influence on the measurement outcomes of the other.
These are the so-called no-signalling constraints. They imply that the probability distribution $p(a,b \vert x,y)$ observed in such an experiment should respect
\begin{eqnarray}
& & p(a\vert x)= \sum_{b} p(a,b \vert x,y)= \sum_{b} p(a,b \vert x,y^{\prime}) \;\; \forall  \, y,y^{\prime} \nonumber\\
& & p(b\vert y)= \sum_{a} p(a,b \vert x,y)= \sum_{a} p(a,b \vert x^{\prime},y) \;\; \forall  \, x,x^{\prime}.
\end{eqnarray}

\begin{figure}[t]
    \centering
    \includegraphics[scale = 0.18]{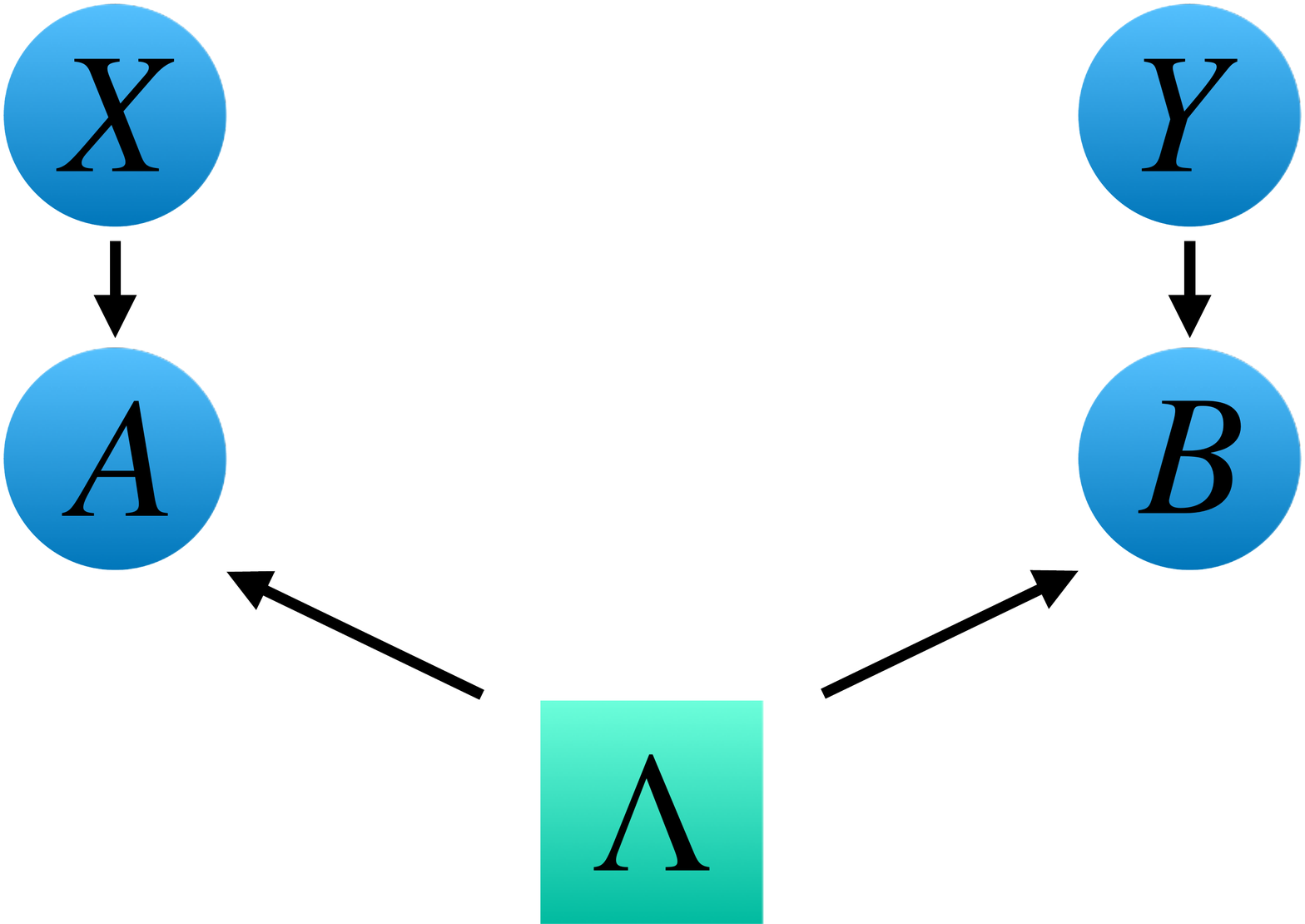}
    \caption{\textbf{Bell's causal structure.} Two distant parties receive their shares of joint subsystem prepared by a common source $\Lambda$. The variables $X$ and $Y$ represent the measurement choices and $A$ and $B$ the corresponding measurement outcomes for the observers Alice and Bob, respectively.}
    \label{fig:Bell}
\end{figure}

From a causal inference perspective \cite{pearl2009causality}, the central question is: what is the simplest causal structure able to recover, in a faithful manner, the probabilistic conditional independence relations corresponding to the no-signalling condition?
That is, these conditional independencies should follow from the causal structure itself and not from a fine-tuning of model parameters. More formally, a causal structure can be defined by a directed acyclic graph (DAG), where the nodes represent the variables of interest and direct edges encode the causal relations among them \cite{pearl2009causality}. Classically, any node $x_i$ in the graph can be understood as a function of its graph-theoretical parents $Pa(x_i)$, implying the causal Markov condition \cite{pearl2009causality}, which stipulates that a variable is conditionally independent of its nondescendents given its parents and thus
\begin{equation}
p(x_1,\dots, x_n)= \prod_{i=1,\dots,n}p(x_i\vert Pa(x_i)) .
\end{equation}

It turns out that~\cite{wood2015lesson} the  observational equivalence class of causal structures which faithfully (that is, without fine tuning of parameters) imply the no-signalling constraints  is the one which includes the causal structure posited by Bell in his seminal work, shown in Fig.~\ref{fig:Bell}\footnote{There are actually \emph{six} distinct latent variable causal structures which faithfully imply the no-signalling constraints; see \citet[Fig.~24]{wood2015lesson}. Those six of those causal structures, however, comprise a single observational equivalence class, per~\citet{Evans2015}. That is, all six causal structures imply the same Bell inequalities.} Employing the causal Markov condition \cite{pearl2009causality} any probability compatible with Bell's causal structure should fulfill
\begin{equation}
\label{eq:LHV}
p(a,b \vert x,y)= \sum_{\lambda}p(a\vert x,\lambda)p(b\vert y,\lambda)p(\lambda). 
\end{equation}

As shown by Bell~\cite{bell1964einstein},
a quantum experiment which mirrors this causal structure can yield
correlations that are incompatible with the classical description given by \eqref{eq:LHV}. This is the phenomenon known as Bell nonlocality and which can be witnessed by violating Bell inequalities, linear constraints of the general form
\begin{equation}
I=\sum_{a,b,x,y}\alpha_{a,b,x,y} p(a,b\vert x,y) \leq L,
\end{equation}
where $L$ is the maximum possible value achievable by the classical description \eqref{eq:LHV}.

Often it is said that Bell's theorem shows the incompatibility between quantum theory and the assumptions of realism, locality and freedom of choice. Under this standard view, the classical decomposition \eqref{eq:LHV} is referred to as a local hidden variable (LHV) model. This interpretation follows naturally from the causal perspective. The realism assumption 
corresponds to the assumption of explainability of the correlations in terms of a classical causal model, in particular, one positing the existence of a hidden variable $\Lambda$. 
In turn, the locality and freedom of choice (or measurement independence) assumptions are encoded in the causal structure of Fig.~\ref{fig:Bell}.  Through the Markov condition, this structure implies that 
${p(a\vert b,y,x,\lambda)=p(a\vert x,\lambda)}$ and ${p(b\vert a,x,y,\lambda)=p(b\vert y,\lambda)}$, a condition that is typically called `local causality'  \cite{jarrett1984physical,wiseman2017causarum} and which asserts that  Alice's and Bob's outcomes are fully determined by their choices and the state of the source.
It also implies that 
$p(x,y,\lambda)=p(x,y)p(\lambda)$, the condition that is typically taken to formalize the notion of  measurement independence and which asserts that Alice and Bob's choices are independent from the common source establishing the correlations.

\subsection{When Measurement Independence Fails}

The failure of measurement independence is called measurement dependence. In a scenario with measurement dependence we have $p(x,y,\lambda)\neq p(x,y)p(\lambda)$. In stark contrast with the causal models of the standard Bell scenario given by Eq.~\eqref{eq:LHV}, in the presence of measurement dependence we allow $p(\lambda\vert x,y)\neq p(\lambda)$. Without measurement independence, the admissible correlations $p(a,b \vert x,y)$ are fairly unrestricted; certainly standard Bell inequalities would no longer constrain them.

However, just because one grants that $X$ or $Y$ might be \emph{somewhat} correlated with $\Lambda$ does not mean that $X$ and $Y$ are not \emph{also} largely functions of causal factors which \emph{are} independent of $\Lambda$. 
To this end, we formally introduce \emph{local laboratory private randomness sources} $U_x$ and $U_y$ for Alice and Bob respectively, as depicted in Fig.~\ref{fig:Bell3}. We now provide some physical intuition for why it makes sense to model $X$ and $Y$ in terms of both causal factors confounded with $\Lambda$ as well as private causal factors independent of  $\Lambda$.

\begin{figure}[t]
    \centering
    \includegraphics[scale = 0.3]{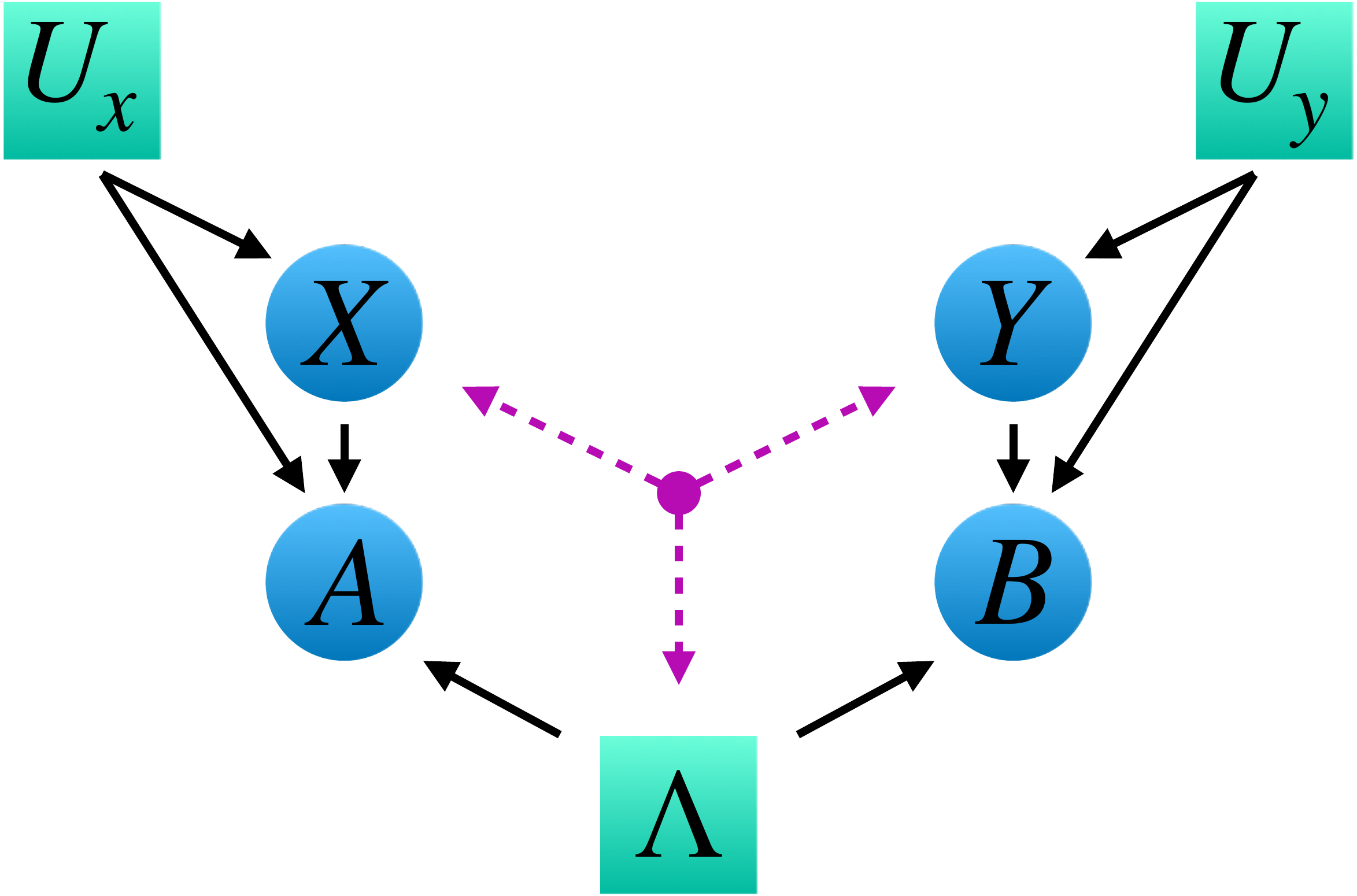}
    \caption{\textbf{Bell's causal structure with measurement dependence.} The purple arrows represent the fact that the correlations between $X$, $Y$ and $\Lambda$ can be due to some direct causal influence among them or mediated by an external common source. We emphasise that \emph{even though the settings are plausibly common-cause connected with $\Lambda$}, typically one can also be confident that they are also (strongly) influenced by independent localized sources of randomness.}
    \label{fig:Bell3}
\end{figure}

Differently from the locality assumption that can be assured by invoking special relativity, the measurement independence is a thorny issue. Nevertheless, the independence of $U_x$, $U_y$ and $\Lambda$ can indeed be made very plausible. For instance, $U_x$ and $U_y$ could stand for stars emitting cosmic photons centuries ago \cite{handsteiner2017cosmic,rauch2018cosmic} or even for the human randomness of hundreds of thousands of people around the globe \cite{big2018challenging}. It seems reasonable that such sources are independent of $\Lambda$, that in a photonic experiment would represent the laser and non-linear crystal employed to generate entangled photons. 

 By explicitly introducing $U_x$ and $U_y$ into our causal models, we can now rephrase the assumption of measurement independence as equivalent to the assumption that $X$ depends \emph{exclusively} on $U_x$ and not on any hidden factor correlated with $\Lambda$, i.e., the assumption that ${p(x\vert u_x,\lambda)=p(x\vert u_x)}$. While is clear that (by definition) $U_x$ is independent of $U_y$ and $\Lambda$, it can be difficult to rigorously justify the assumption of the independence of $X$ and $Y$ from $\Lambda$.

 As an illustration for why this is the case, suppose that $U_x$ represents some far away star emitting cosmic photons that define the variable $X$. There is no reason to doubt that this star is independent of a laser in a laboratory today. However, at some point, the photons emitted by these independent sources (the star and the laser) meet up within the physicist's lab and define the outcome $A$ they will give rise to. Within this context, the atmosphere in the lab (acting as a medium for the photons) or whatever else that might affect the photons state can act as a source of correlations and lead to deviations from perfect
 measurement independence.

In short, even though the experiment might use a source $U_x$ that is independent of  $\Lambda$, a causal mediary between $U_x$ and the measurement outcome variable $A$ might nonetheless be influenced by something that also has an influence on $\Lambda$, so that $X$ and $\Lambda$ end up having a common cause (as illustrated in Fig.~\ref{fig:Bell3}) and therefore the potential for small amount of statistical dependence between them.
Furthermore, as it turns out, even a quite small amount of measurement dependence (to be explicitly quantified below) is already sufficient to simulate the maximum possible Bell inequality violation achievable with quantum mechanics \cite{hall2010local,hall2020measurement}.

Even though the sources $U_x$, $U_y$ and $\Lambda$ can be assumed independent, the most general local hidden variable model represented by the DAG in Fig.~\ref{fig:Bell3} implies that
\begin{widetext}
\begin{align}
\begin{split}
\label{eq: p(a,b|x,y)}
p(a,b,x,y) &= \sum_{\lambda} \Bigg(\sum_{u_x}p(a\vert x,u_x,\lambda)p(x\vert u_x,\lambda)p(u_x)\Bigg) \Bigg(\sum_{u_y}p(b\vert y,u_y,\lambda) p(y \vert u_y,\lambda)p(u_y)\Bigg) p(\lambda)
\\ & = \sum_{\lambda} p(a\vert x,\lambda)p(x\vert \lambda) p(b\vert y,\lambda) p(y \vert \lambda) p(\lambda)
\end{split}\shortintertext{such that}
\begin{split}
\label{eq:lhvMD} p(a,b \vert x,y) &= \sum_{\lambda} \Bigg(\sum_{u_x}p(a\vert x,u_x,\lambda)p(u_x\vert x,\lambda)\Bigg) \Bigg(\sum_{u_y}p(b\vert y,u_y,\lambda) p(u_y\vert y,\lambda)\Bigg) p(\lambda\vert x,y)
\\ & = \sum_{\lambda} p(a\vert x,\lambda)p(b\vert y,\lambda) p(\lambda\vert x,y)
\end{split}\end{align}
\end{widetext}
Note that $p(\lambda\vert x,y)= p(\lambda)$ \emph{only} if $p(x\vert u_x,\lambda)=p(x\vert u_x)$ and $p(y\vert u_y,\lambda)=p(y\vert u_y)$, in which case then we recover the standard measurement independent model of Eq. \eqref{eq:LHV}. 

\subsection{Quantifying Measurement Dependence}

In the special limit $p(x,y,\lambda)- p(x,y)p(\lambda)\to 0$ we recover measurement independence, so one natural measure for quantifying the degree of measurement dependence in a given Bell experiment is \cite{chaves2015unifying}
\begin{equation}
\label{eq:m1}
\mathcal{M}\coloneqq \sum_{x,y,\lambda} \vert p(x,y,\lambda)- p(x,y)p(\lambda) \vert.    
\end{equation}
Clearly, when $\mathcal{M}=0$, the measurement-dependent model \eqref{eq:lhvMD} goes over to the usual measurement-independent model \eqref{eq:LHV}.

A related measure that has also been considered in the literature \cite{hall2010local,hall2020measurement} is the mutual information between the experimenter's choices and the source $\Lambda$, defined as
\begin{equation}
\label{eq:mi}
    I(X,Y:\Lambda)\coloneqq H(X,Y)+H(\Lambda)-H(X,Y,\Lambda),
\end{equation}
where $H(X)=-\sum_{x}p(x)\log{p(x)}$ is the Shannon entropy of the random variable $X$. Importantly, the measures \eqref{eq:m1} and \eqref{eq:mi} can be related via the Pinsker inequality \cite{fedotov2003refinements}
\begin{equation}\label{eq:Pinsker}
 \mathcal{M}^2 \leq \frac{I(X,Y:\Lambda)}{ \log_2 e}.
\end{equation}

Within this context, we can then ask how much measurement dependence would be necessary to reproduce some Bell inequality violation using the classical causal model described by \eqref{eq:lhvMD} (see Fig. \ref{fig:Bell3}). For instance, in the paradigmatic Clauser-Horne-Shimony-Holt (CHSH) scenario \cite{clauser1969proposed}, where each of the parties can measure two possible dichotomic measurements,
LHV models~\eqref{eq:LHV} respect the inequality
\begin{equation}
    \CHSH = \langle A_0B_0 \rangle+\langle A_0B_1 \rangle+\langle A_1B_0 \rangle-\langle A_1B_1 \rangle \leq 2,
\end{equation}
where $\langle A_xB_y \rangle= \sum_{a,b=0,1} (-1)^{a+b}p(a,b\vert x,y)$ is the expectation value of Alice and Bob's outcomes $A$ and $B$ conditioned on the inputs $x$ and $y$ (with all input and output variables taking values $0$ or $1$).
For measurement-dependent models \eqref{eq:lhvMD}, one can prove that \cite{chaves2015unifying}
\begin{equation}
\label{eq:CHSHM}
(\CHSH-2)/4 \leq \mathcal{M}
\end{equation}
In turn, using the mutual information, it has been proven that  \cite{hall2020measurement}
\begin{equation}
\label{eq:CHSHI}
2 -h\left(\frac{4-\CHSH}{8}\right) - \frac{4+\CHSH}{8}\log_23 \leq I(X,Y:\Lambda)  ,
\end{equation}
where $h(x)$ is the binary entropy given by
\begin{equation}
h(x)=-x\log_2 x-(1-x)\log_2(1-x) .  
\end{equation}
This implies, in particular, that even for the maximal quantum violation of $2\sqrt{2}$ of the CHSH inequality, a measurement dependence as low as  $I(X,Y:\Lambda)=0.046$  bits is already enough to reproduce the quantum predictions.

Equations \eqref{eq:CHSHM} and \eqref{eq:CHSHI} show that
the violation of a Bell inequality can be explained by
some non-zero degree of measurement dependence. The problem, however, is that whether or not there is such measurement dependence cannot be determined
in a standard Bell experiment. In other words, one cannot know if the violation of a given Bell inequality was due to quantum effects or simply because the measurement independence assumption was not fulfilled. As such, the measurement dependence loophole remains as a possible classical explanation for the violation of a Bell inequality, even for milestone experiments violating a Bell inequality while sealing the locality and detection efficiency loopholes~\cite{Shalm,Giustina,Hensen}.

\begin{figure}[t]
    \centering
    \includegraphics[scale = 0.30]{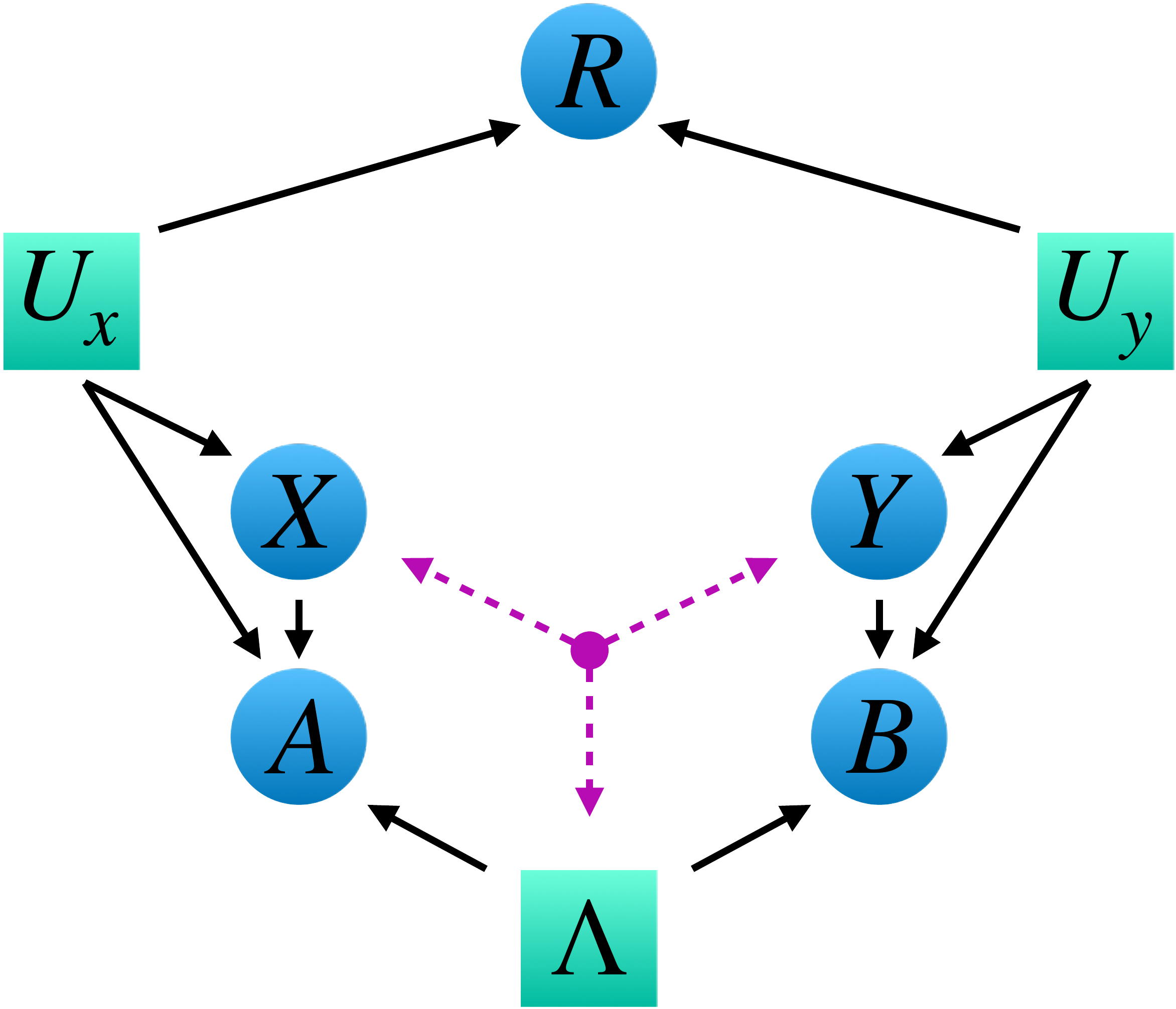}
    \caption{\textbf{Causal structure assessing the measurement dependence.} The extra measurement outcome variable $R$ can be used to provide an upper bound on the measurement dependence. Notice that this model allows, at least in principle, for arbitrary correlations between the input variables $X$ and $Y$ and the source $\Lambda$.}
    \label{fig:Bell4}
\end{figure}

At the core of the problem is the fact that, in a standard Bell experiment, there is nothing that implies an {\em upper bound} on the amount of measurement dependence.

In the following, we will show that by embedding a Bell experiment in a larger causal network that includes an auxiliary variable $R$ which is influenced by $U_x$ and $U_y$ but is independent of $\Lambda$ (see Fig.~\ref{fig:Bell4}), we are able to derive an upper bound on the amount of measurement dependence. In essence, the idea is that, for this causal structure, if $R$ is strongly correlated with $X$, then $X$ must be only weakly correlated with $\Lambda$, and similarly for $Y$. (The limiting case of this tradeoff---the fact that {\em perfect} correlation between $R$ and $X$ implies {\em no} correlation between $X$ and $\Lambda$---is what Fritz~\cite{fritz2012beyond} used to establish that the triangle causal scenario can be mapped onto the Bell scenario. The latter observation is therefore the root of the approach to bounding measurement dependence described in this article.)

This fact allows us to answer in an unambiguous manner (assuming the sources $U_x$, $U_y$ and $\Lambda$ to be independent) the question of whether the violation of a Bell inequality is due to the presence of quantum entanglement or due to measurement dependence. We will also extend our analysis to the multipartite case and show how our inequalities can be used to characterize a large class of causal networks that are increasingly attracting attention \cite{branciard2010characterizing,branciard2012bilocal,fritz2012beyond,fritz2016beyond,tavakoli2014nonlocal,chaves2016polynomial,poderini2020experimental,rosset2016nonlinear,andreoli2017maximal,chaves2018quantum,renou2019genuine,carvacho2017experimental}. Prior to doing so, however, we briefly introduce the entropic framework for causal inference \cite{chaves2014inferring} that will be crucial to derive some of our technical results.

\section{Entropic inequalities and causal networks}\label{sec:entropic}

Our aim is to obtain an upper bound to the measurement dependence measures, $\mathcal{M}$ or $I(X,Y: \Lambda)$ in terms of the degree of correlation exhibited between $R$ and $X,Y$ in the causal structure of Fig.~\ref{fig:Bell4}.
Note that the joint distributions on observed variables that are compatible with the causal structure of Fig.~\ref{fig:Bell4} are of the form:
\begin{align}\label{eq:triangleMD}
p(a,b,x,y,r) = \smashoperator[l]{\sum_{u_x,u_y,\lambda}}\begin{pmatrix}
p(a,x\vert u_x,\lambda)p(b,y\vert u_y,\lambda)\times\\
 p(r\vert u_x,u_y)p(u_x)p(u_y)p(\lambda)
\end{pmatrix}.
\end{align}
However, the fact that the sources $U_x$, $U_y$ and $\Lambda$ are independent 
implies that this set 
is nonconvex and therefore difficult to characterize
\cite{wolfe2019inflation,chaves2016polynomial}. This nonconvexity can be circumvented by the entropic approach introduced in \cite{fritz2012entropic,chaves2014causal,chaves2014inferring,chaves2015information}, allowing one to obtain analytical bounds for $I(X,Y:\Lambda)$. The bounds can then also be translated into bounds on the L1-norm quantifier $\mathcal{M}$ via the Pinsker inequality~\cite{fedotov2003refinements}.

A detailed account of the entropic approach can be found at \cite{chaves2014inferring}. Here we will introduce the central concepts necessary to understand the results that will follow.

Consider a set of $n$ discrete random variables $X_1, \dots,
X_n$. We denote as
$[n]=\{1, \dots, n\}$ the set of indices of these random variables. For every subset $S\in 2^{[n]}$ of indices, $X_S$ is
the random vector $(X_i)_{i\in S}$ and $H(S):=H(X_S)$ is its associated Shannon entropy defined by ${H(X_S)\coloneqq-\sum_{x_s}p(x_s)\log_2 p(x_s)}$. We can construct an entropy vector $h=\left\{\emptyset, H(X_1),H(X_2), H(X_1,X_2),\dots,H(X_1,\dots,X_n) \right\}$ with all possible $2^n$ entropies for $n$ variables (including the empty set) and ask what are the constraints for $h$ to be a valid entropy vector. The region of the real space $\mathbb{R}^{2^n}$ corresponding to entropies  is known to define a convex cone \cite{yeung2008information}; a complete and explicit description remains unknown. For this reason, one has to work with an outer approximation, known as the Shannon cone $\Gamma_n$, defined by the set of linear inequalities given by
\begin{subequations}
\begin{align}
	\label{eq:monotonicity} H_{}([n]\setminus\{i\}) &\leq H_{}([n])\\
	\label{eq:subadditivity} H_{}(S) + H_{}(S\cup\{i,j\}) &\leq H_{}(S\cup \{i\}) + H_{}(S\cup \{j\})
	\\\label{eq:pointdist} H_{}(\emptyset) &= 0 
\end{align}
\end{subequations}
for all $S \subset [n] \setminus\{i,j\}$, $i \neq j$ and $i, j\in
[n]$.  These inequalities are known as the elementary
inequalities and any inequality that follows from
the elementary set is said to be of the Shannon-type. The first constraint~\eqref{eq:monotonicity} is known as monotonicity and
states that the uncertainty about a set of variables should always be larger than or equal to the
uncertainty about any subset of it, i.e., nonnegativity of conditional entropy. The second constraint~\eqref{eq:subadditivity} is called strong sub-additivity and is equivalent to the nonnegativity of the conditional mutual information. That is, 
\begin{align*}
I(&X_i:X_j\vert X_S)\coloneqq
\\& H(X_{S \cup i}) +H(X_{ S \cup j})-H( X_{S
\cup\{i,j\}})-H(X_{S})\\
&\geq 0\,.
\end{align*}

The causal relations implied by a given causal structure can be easily integrated in this framework as linear constraints. For  instance, the independence of the sources in the causal structure of Fig.~\ref{fig:Bell4}
implies that  $H(U_x,U_y,\Lambda)=H(U_x)+H(U_y)+H(\Lambda)$. The subspace of $\mathbb{R}^{2^n}$
defined by all such constraints can be denoted as $\Gamma_c$. Thus, any entropy vector compatible with a given causal structure should lie in the convex cone $\Gamma_n^c:=\Gamma_n \cap \Gamma_c$. Since the sources are not directly observable in the Bell experiment, they need to be traced out from our description, an instance of a quantifier elimination problem that in the entropic case can be performed by a simple Fourier-Motzkin elimination \cite{williams1986fourier}.

\section{Bounding the measurement dependence
}\label{sec:nonlinearineqs}

As noted earlier, in order to upper bound the measurement dependence $I(X,Y:\Lambda)$, we modify the causal structure to that of Fig.~\ref{fig:Bell4}, wherein there is an extra
variable $R$ that might depend on the sources $U_x$ and $U_y$ but is independent of $\Lambda$.

Employing the general entropic framework introduced in \cite{chaves2014inferring} and outlined above, we can completely characterize the Shannon inequalities bounding the measurement dependence ${I(X,Y:\Lambda)}$. For our purpose, the causal constraints implied by the DAG  in Fig.~\ref{fig:Bell4} can be summarized by the entropic constraints 
\begin{subequations}
\begin{align}\label{eq:source_indep_entropy}
&H(U_x,U_y,\Lambda)=H(U_x)+H(U_y)+H(\Lambda),\quad\text{and}
\\\label{eq:zero_mut_info} 
\begin{split}
&I(R:X,Y,\Lambda\vert U_x,U_y) = 0, \\
&I(X:R,Y,U_y\vert U_x,\Lambda) = 0, \\
&I(Y:R,X,U_x\vert U_y,\Lambda) =0,
\end{split}
\end{align}
\end{subequations}
Eq.~\eqref{eq:source_indep_entropy} follows from the independence of the sources while Eqs.~\eqref{eq:zero_mut_info} encode the zero conditional mutual information between any random variable and its causal nondescendants given its parents, i.e., the local Markov condition.

Using the approach delineated before and performing the corresponding Fourier-Motzkin elimination~\cite{williams1986fourier}, we find three non-trivial upper bounds for $I(X,Y:\Lambda)$. 
\begin{lemma}\label{lem:up1}
For any data compatible with the classical causal structure in Fig.~\ref{fig:Bell4}, we find that ${I(X,Y:\Lambda)\leq \Theta(X,Y,R)}$, where 
\begin{align}
\Theta(&X,Y,R)\coloneqq 
\\&\min \begin{cases} \begin{smallmatrix} H(X,Y \vert R)\,,\end{smallmatrix} \\ 
\begin{smallmatrix} H(X,Y) -I(X:Y:R)
   -I(X:R) -I(Y:R)\,,
\end{smallmatrix}\\
\begin{smallmatrix} 
H(X,Y)+H(R) -2I(X:Y:R)-2I(X:R)  -2I(Y:R)\,,
\end{smallmatrix}
\end{cases}\nonumber
\end{align} 
and where the term $I(X:Y:R)$ is the tripartite mutual information, which can be rewritten as \footnotesize{$I(X:Y:R)\coloneqq H(X,Y,R)-H(X,Y)-H(X,R)-H(Y,R)+H(X)+H(Y)+H(R)$}.
\end{lemma}

The entropic approach also gives rise to the lower bound given by
\begin{eqnarray}
I(X:Y) \leq I(X,Y:\Lambda).
\end{eqnarray}

Each of the upper bounds in Lemma~\ref{lem:up1} can combined with~\eqref{eq:CHSHM} or~\eqref{eq:CHSHI} to give rise to a non-linear Bell inequality. 
\begin{prop}For observational data compatible with the classical causal structure in Fig.~\ref{fig:Bell4} we find that
\begin{equation}
\label{eq:ineqCHSH1}
2 -h\left(\frac{4-\CHSH}{8}\right) - \frac{4+\CHSH}{8}\log_23  \leq \Theta(X,Y,R)\,,\\
\end{equation}
by virtue of combining~\eqref{eq:CHSHI} with Lemma~\ref{lem:up1}, as well as
\begin{equation}
\label{eq:ineqCHSH2}
\frac{\CHSH-2}{4} \leq \sqrt{\frac{\Theta(X,Y,R)}{\log_2 e}}\,,
\end{equation}
by virtue of combining \eqref{eq:CHSHM} with Lemma~\ref{lem:up1} through the Pinsker inequality~\eqref{eq:Pinsker}.
\end{prop}

Under the assumption that the sources $U_x$, $U_y$ and $\Lambda$ are independent, a violation of any combination of these inequalities would mean that whatever degree of
measurement dependence is present,
i.e., whatever value $I(X,Y:\Lambda)$ takes,
it is not enough to explain the observed correlations. Thus, we would be unambiguously witnessing nonclassicality.
Notice that if the input variables $X$ and $Y$ are perfectly correlated with the auxiliary variable $R$, then $H(X,Y\vert R)=0$ implying that $I(X,Y:\Lambda)=0$. In this case, we recover the usual Bell scenario with no measurement dependence. It is important to highlight, however, that in our scenario, an upper bound on the amount of
measurement dependence is 
implied
by the empirical data (assuming the independence of sources) and not assumed a priori, like in a standard Bell scenario.

Notice that the upper bounds in  Lemma~\ref{lem:up1} are valid for an arbitrary number of inputs and outputs. Thus, inequalities like \eqref{eq:ineqCHSH1} and \eqref{eq:ineqCHSH2} can be derived for arbitrary bipartite Bell scenarios. To illustrate, in \cite{chaves2015unifying} it has been noticed that the measure $\mathcal{M}$ can also be related to the CGLMP inequality \cite{PhysRevLett.88.040404}, a Bell inequality bounding classical correlations in a scenario where Alice and Bob have $d$ possible outcomes. More precisely,
\begin{equation}
\label{eq:CGLMM}
(I_d-2)/4 \leq \mathcal{M}
\end{equation}
where the CGLMP inequality is given by \cite{PhysRevLett.88.040404}
\begin{eqnarray}
I_d\coloneqq & & \sum_{k=0}^{{\lceil d/2 \rceil}-1}\left(1- \frac{2k}{d-1} \right) \\ \nonumber 
& &\big[ p(a_0=b_0+k)+p(b_0=a_1+k+1) \\ \nonumber
& & +p(a_1=b_1+k)+p(b_1=a_0+k) \\ \nonumber
& & - p(a_0=b_0-k-1)-p(b_0=a_1-k) \\ \nonumber 
& & -p(a_1=b_1-k-1)-p(b_1=a_0-k-1) \big] \leq 2
\end{eqnarray}
and where 
\begin{equation}
p(a_x=b_y+k)=\sum_{j=0}^{d-1}p(a_x=j,\ b_y=j+k \quad \mathrm{mod} \quad d), 
\end{equation}
with $p(a_x,b_y)=p(a,b\vert x,y)$. Using the Pinsker inequality~\eqref{eq:Pinsker}, we can readily derive a generalization of inequality \eqref{eq:ineqCHSH2}.
\begin{prop}
For observational data compatible with the classical causal structure in Fig.~\ref{fig:Bell4}, we find that 
\begin{equation}
\label{eq:ineqCGLMP}
\frac{I_d-2}{4}
\leq \sqrt{\frac{\Theta(X,Y,R)}{\log_2 e}}.
\end{equation}
\end{prop}
Violation of inequality~\eqref{eq:ineqCGLMP} implies that the degree of violation of the CGLMP inequality cannot be accounted for by measurement dependence and therefore attests to the presence of nonclassicality.

\subsection{Example: The Fritz distribution}

To illustrate our results we consider the Fritz distribution \cite{fritz2012beyond}, the first known example connecting causal networks with Bell's theorem.  In this case, all variables are binary and the measurement outcome of variable $r$ consists of two bits, $r=(r_0,r_1)$. As argued by Fritz, if the bit $x$ is perfectly correlated with $r_0$, this implies that $x$ should be completely uncorrelated from the source $\lambda$. Similarly, perfect correlation between $y$ and $r_1$ implies that $y$ is uncorrelated from the source $\lambda$. That is, the variables $X$ and $Y$ can be seen as the standard measurement choices of Alice and Bob in the usual Bell scenario. Under this condition of perfect correlations, the violation of a standard Bell inequality by the conditional distribution $p(a,b \vert x,y)$ is then a sufficient condition to witness the nonclassicality. 

A quantum realization of such scenario is given by
\begin{align}\label{qexpress}
p(a,x,b&,y,r_0,r_1)
\\\nonumber =\mathrm{Tr}&\left({\rho_{AB}} \otimes {\rho_{X R_0}} \otimes {\rho_{Y R_1}}\,\cdot \,  {M^{A X}_{a,x}} \otimes {M^{B Y}_{b,y}} \otimes {M^{R_0 R_1}_{r_0,r_1}} \right)\;,
\end{align}
where $\rho_{A B}$ denotes the density operator of the state shared between Alice and Bob (thus replacing the classical description in terms of the hidden variable $\Lambda$) and similarly for $\rho_{X R_0}$ and $\rho_{Y R_1}$; $\{M^{A X}_{a,x}\}$ denotes a POVM acting on the physical system in possession of Alice (similarly for $\{M^{B Y}_{b,y}\}$ and $\{M^{R_0 R_1}_{r_0,r_1}\}$). In the Fritz case, the sources $\rho_{A B}$, $\rho_{X R_0}$, $\rho_{Y R_1}$ are given by three singlet states $\ket{\Phi}=(1/\sqrt{2})(\ket{00}+\ket{11}$), and the POVMs have the following form:
\begin{align}\begin{split}\label{eq:FritzPOVMs}
& M^{R_0 R_1}_{(r_0,r_1)}=M^{R_0}_{r_0}\otimes M^{R_1}_{r_1}, \\
&
M^{A X}_{(a,x)}=M^{X}_{x}\otimes M^{A}_{a|x},
\\
&
M^{B Y}_{(b,y)}=M^{y}_{y}\otimes M^{B}_{b|y},
\end{split}
\end{align}
where $\{ M^{R_0}_{r_0}\},$ $\{M^{R_1}_{r_1}\},$ $\{M^{X}_{x}\},$ $\{M^{Y}_{y}\}$ are all 
measurements of the $\sigma_z$ basis, $ \{ M^{A}_{a|x}\}$ corresponds to one of the two Pauli observables among $\{\sigma_x, \sigma_z\}$ depending on the value of $x$, and $\{ M^{B}_{b|y} \}$ corresponds to one of the two observables among  $\{{(\sigma_z+\sigma_x)/\sqrt{2}}, {(\sigma_z-\sigma_x)/\sqrt{2}}\}$  depending on the value of $y$. The fact that the measurements in Fritz's example have been chosen to ensure that the conditional $p(a,b \vert x,y)$ violates a Bell inequality implies that Fritz's distribution $p(a,x,b,y,r_0,r_1)$ has no classical explanation.

Any experiment that aims to realize the Fritz distribution in the triangle scenario aims to realize the ideal states and measurements specified above, but due to the inevitability of noise, the states and measurements that are actually implemented are necessarily noisy versions of these. Consider, for instance, that the source states are noisy versions of the Bell state, given by $\rho=v\ket{\Phi}\bra{\Phi}+(1-v)\openone/4$. This implies that the correlations between $x$ and $r_0$ and between $y$ and $r_1$ will not be perfect and the Fritz argument cannot be employed any longer. Even though in this case we do not have any measurement dependence, the point is that the correlations generated by such model are indistinguishable from a measurement dependent model. In other terms, to be sure about the non-classicality of the data we have to employ the causal network delineated above. For this case, however, $\Theta(X,Y,R)$ is given by 
\begin{align}
&2-\frac{\operatorname{s}\left((v-1)^2\right)+\operatorname{s}\left((v+1)^2\right)+\operatorname{s}\left(1-v^2\right)}{4}\,, 
\\\nonumber &\text{where}\qquad \operatorname{s}(x)\coloneqq x\;\; log_2(x)\,,
\end{align}
implying that visibilities as high as $v\approx 0.994$ are required to observe any violation of the inequalities~\eqref{eq:Pinsker}~or~\eqref{eq:CHSHI} and thus witness non-classicality even in the potential presence of measurement dependence. It is worthy to point out, however, that the source shared between $X$ and $R_0$ and the source between $Y$ and $R_1$ do not need to have a quantum nature. Since we are simply measuring such states in the computational basis, the same correlations can be achieved with a classical source, significantly simplifying an experimental test.

In hindsight, it is not surprising that the inequalities we derive are not robust. It is known that measurement dependence is a very strong resource to simulate nonlocal correlations in a Bell scenario \cite{hall2010local,hall2020measurement}. Remarkably, however, different approaches that do not hinge on Bell's theorem can tolerate a significant amount of measurement dependence, way beyond what is possible within the standard Bell scenario. As will be explored in details elsewhere, resorting to the inflation technique using the Web inflation \citep[Fig.~2]{wolfe2019inflation} we can derive new non-linear inequalities allowing for visibilities as low as $v\approx 0.907$, indeed showing that the causal network we propose here not only leads to testable constraints but can also tolerate much more measurement dependence than usual approaches.

\subsection{Multipartite Bell Inequalities without Measurement Independence}

So far, we have focused on the bipartite scenario, but our results can be readily extended beyond this.
Concretely, consider the case of $n$ parties, each $i$-th part with an input $X_i$ and output $A_i$; the cardinalities of the inputs and outputs being arbitrary (see Fig. \ref{fig:Bell5}).

Similarly to the bipartite case, we introduce an auxiliary variable $R$ that can depend on all the sources of local laboratory private randomness $\{U_i\}_i$ where $U_i$ accounts for causal influences on $X_i$ and $A_i$ which are independent of $\Lambda$ (with ${i=1,\dots,n}$) (see Fig.~\ref{fig:Bell6}). 
The joint distributions over the observed variables that are compatible with the causal structure of Fig.~\ref{fig:Bell6} are:

\begin{align}\label{eq:multiBellform}
&p(a_1,...,a_n,x_1,...x_n,r)=\nonumber\\
 &\smashoperator[l]{\sum_{u_1,...,u_n,\lambda}}
p(r\vert u_1,...,u_n)p(\lambda)\prod_{i=1}^n p(a_i,x_i\vert u_i,\lambda)p(u_i)
\end{align}
Observational statistics over the original observable variables together with $R$ can then be employed to upper bound $I(X_1,\dots,X_n:\Lambda)$.
\begin{lemma}\label{lem:generalub}
For any data compatible with the classical causal structure in Fig.~\ref{fig:Bell6}, we find that \begin{align}
I(X_1,\dots,X_n:\Lambda) \leq H(X_1,\dots,X_n\vert R).
\end{align}
\end{lemma}\begin{proof}
To prove Lemma~\ref{lem:generalub}, we collect all Bell scenario inputs into the composite random variable ${\boldsymbol{X}=(X_1,\dots,X_n)}$ and all Bell scenario outputs into the composite random variable ${\boldsymbol{A}=(A_1,\dots,A_n)}$. We then combine the Shannon-type inequality
\begin{align}
H(R,\Lambda)+H(\boldsymbol{X})-H(\boldsymbol{X},\Lambda) \leq H(\boldsymbol{X}\vert R)+H(R)
\end{align}
with the minimal causal assumption that $R$ is independent of $\Lambda$ such that 
\begin{align}\label{eq:mincausalassumption}
H(R,\Lambda)=H(R)+H(\Lambda)
\end{align}
to obtain~Lemma~\ref{lem:generalub} via the substitution ${H(\Lambda)+H(\boldsymbol{X})-H(\boldsymbol{X},\Lambda)}=I(\boldsymbol{X}:\Lambda)$.\end{proof}

\begin{figure}[t]
    \centering
    \includegraphics[scale = 0.30]{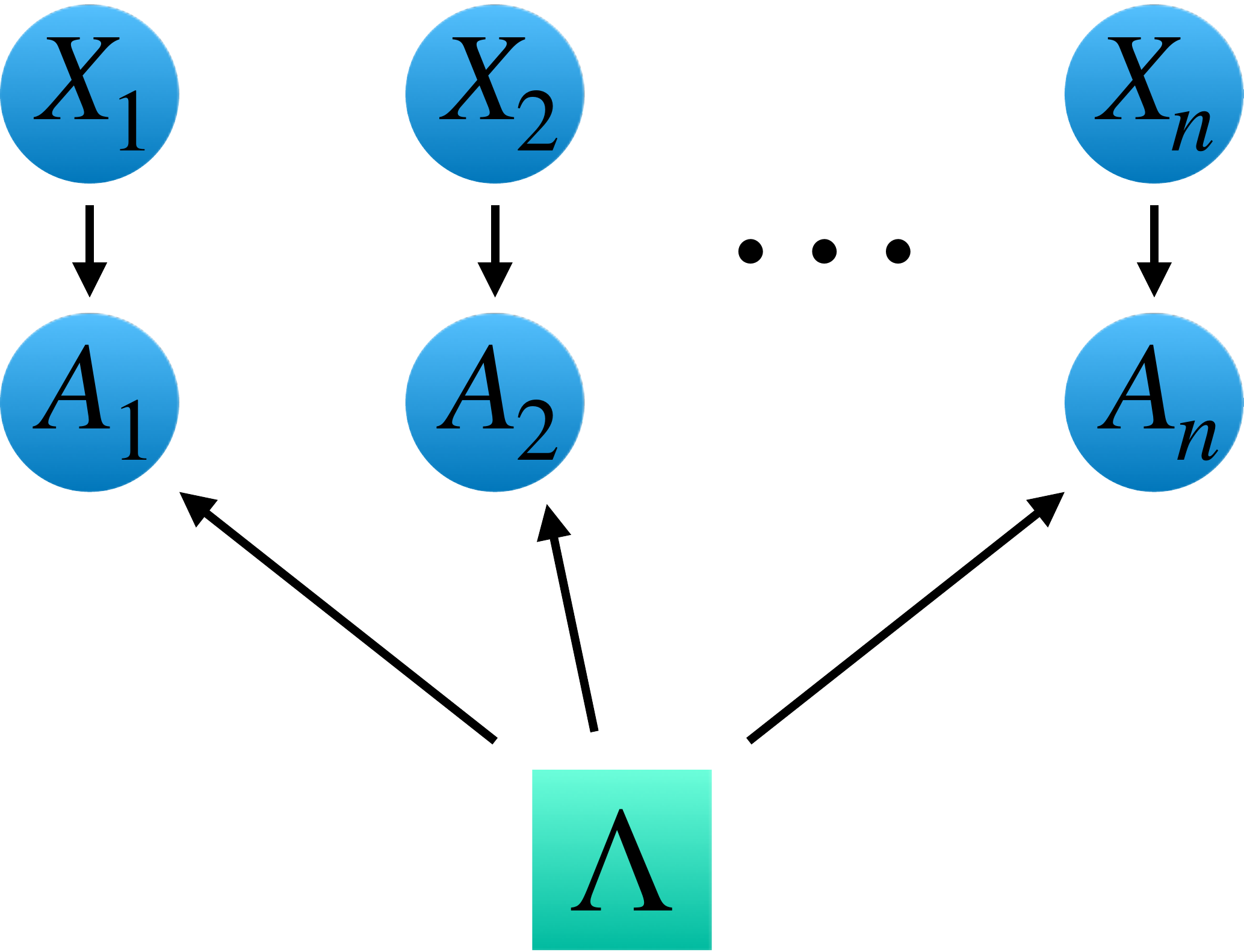}
    \caption{\textbf{Causal structure of a standard multipartite Bell scenario.
    } Each of the $n$ distant parties have a common source $\Lambda$ and inputs and outputs labelled as $X_i$ and $A_i$, respectively.}
    \label{fig:Bell5}
\end{figure}

\begin{figure}[b]
    \centering
    \includegraphics[scale = 0.29]{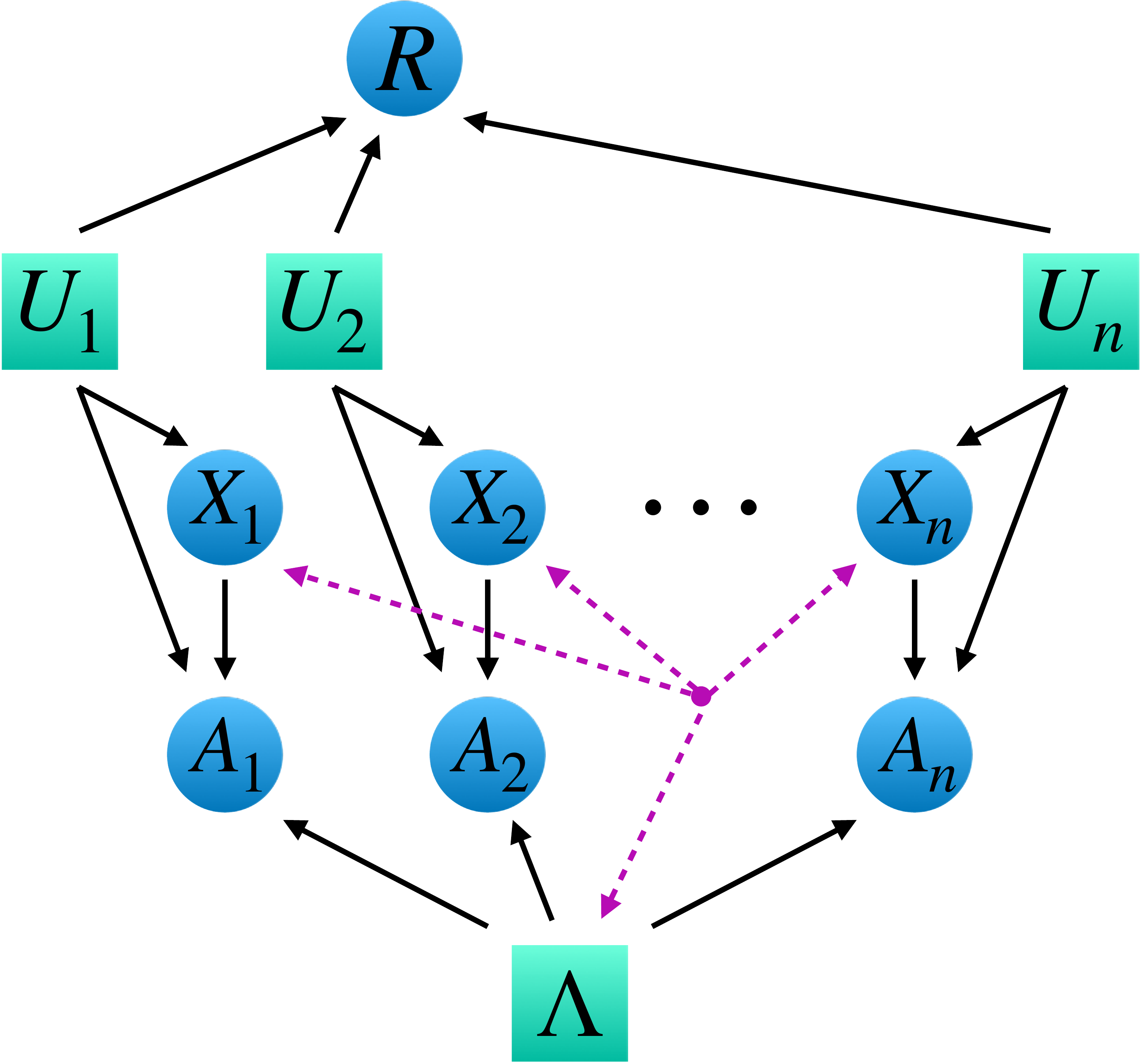}
    \caption{\textbf{Causal structure assessing measurement dependence in a multipartite Bell scenario.}
     The auxiliary variable $R$ allows one to practically upper bound the amount of measurement dependence that might be present. The double arrowed (purple) edges indicate that the correlations between the inputs variable $X_i$ and $\Lambda$ might arise from direct causal influence or via a common source.}
    \label{fig:Bell6}
\end{figure}

Suppose that we have a generic result which relates the violation of some Bell-type function $I$ to the L1-norm quantifier $\mathcal{M}$, i.e. 
\begin{align}
\label{eq:generic}
f(I) &\leq \mathcal{M},\quad\text{where}
\\\nonumber \mathcal{M}&\coloneqq \sum_{x_1 \dots x_n,\lambda}\vert p(x_1,\dots,x_n,\lambda)-p(x_1,\dots,x_n)p(\lambda) \vert ,
\end{align}
and where 
$f$
is some linear function for which $f(I)\leq 0$ for any correlations not violating the Bell inequality (with $f(I) >0$ otherwise). 
Combining~\eqref{eq:generic} with Lemma~\ref{lem:generalub} via the Pinsker inequality~\eqref{eq:Pinsker}, we then generically obtain
\begin{equation}
\label{eq:generalbi}
 f(I)\leq \sqrt{\frac{H(X_1,\dots,X_n\vert R)}{\log_2 e}}.
\end{equation}
Clearly, if $H(X_1,\dots,X_n\vert R)=0$ we recover the usual measurement independent case characterized by the Bell inequality  $f(I)\leq 0$.

Instead of using the Pinsker inequality to connect the L1-norm quantifier $\mathcal{M}$ with the information theoretical measure $I(X_1,\dots,X_n:\Lambda)$, we can try to derive lower bounds for the latter by exploring multipartite inequalities. 

This is how we achieve the following proposition, for the tripartite Bell scenario, by adapting the results in \cite{hall2010local} to the Mermin-Ardehali–Belinski–Klyshko \cite{mermin1990,ardehali1992bell,belinskiui1993} inequality
\begin{align}\nonumber
M &\coloneqq \langle A_0B_0C_1\rangle + \langle A_0B_1C_0\rangle + \langle A_1B_0C_0\rangle - \langle A_1B_1C_1\rangle 
\\&\leq 2.
\end{align}
\begin{prop}\label{prop:multiMDexample}
For observational data compatible with the classical causal structure in Fig.~\ref{fig:Bell6} specialized to the case of three parties, we find that 
\begin{align}\label{eq:multiMDexample}
1 - \frac{1}{2}h\left(\frac{4-M}{8}\right) - \frac{4+M}{16}\log_23 \leq H(X,Y,Z\vert R)
\end{align}
whenever the distribution over the inputs is uniform, i.e., when  $p(x,y,z)=\frac{1}{8}$.
\end{prop}
\begin{proof}
As we show in Appendix \ref{app: Mermin}, assuming $p(x,y,z)=\frac{1}{8}$, we get
\begin{equation}
\label{eq: I_mermin}
1 - \frac{1}{2}h\left(\frac{4-M}{8}\right) - \frac{4+M}{16}\log_23  \leq  I(X,Y,Z:\Lambda) .
\end{equation}
which is the same for signaling and nonsignaling behaviours (see Appendix~\ref{app: nosig}). 
We then obtain Prop.~\ref{prop:multiMDexample} by combining~\eqref{eq: I_mermin} with Lemma~\ref{lem:generalub}.\end{proof}

\begin{figure}[t]
    \centering
    \includegraphics[scale = 0.6]{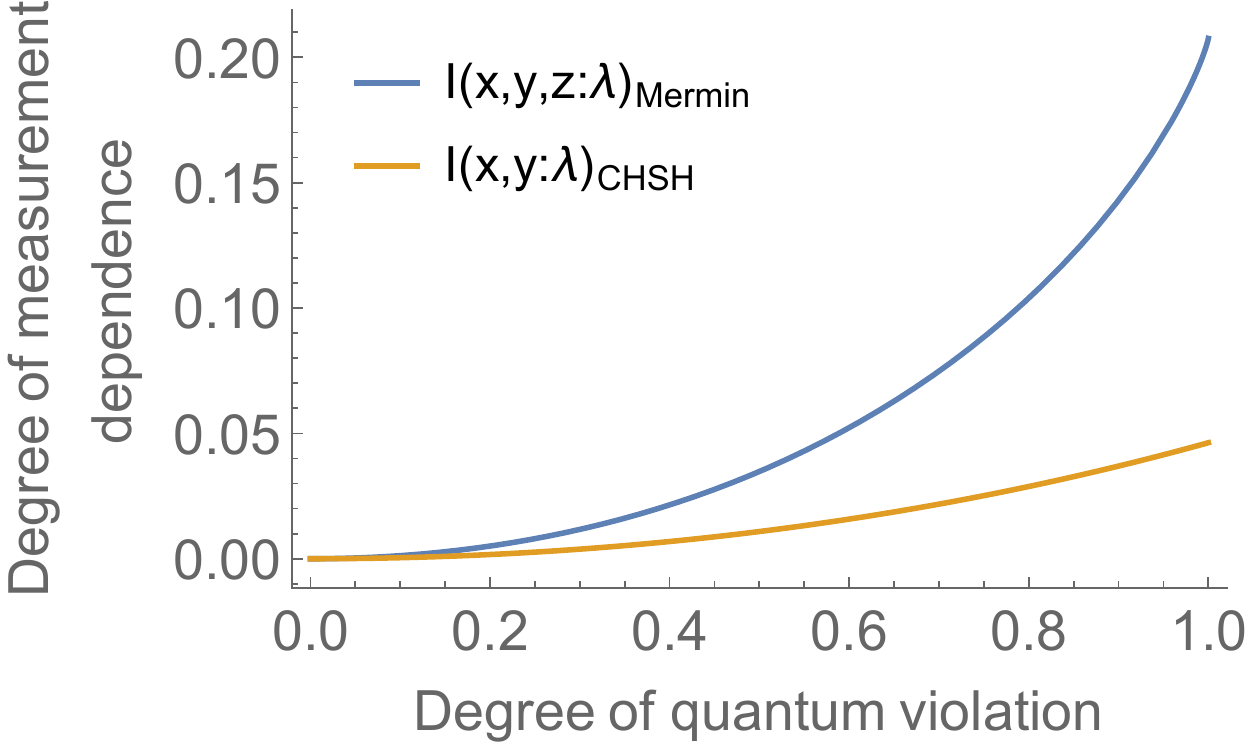}
    \caption{\textbf{Measurement dependence and Bell inequality violations} The orange curve shows the lower bound on the measurement dependence $I(X,Y:\Lambda)$ as described by eq. \eqref{eq:CHSHI} needed to explain a given degree of quantum violation for the CHSH inequality given by $\frac{\mathrm{CHSH}-2}{2\sqrt{2}}$. The blue curve shows the lower bound on the measurement dependence $I(X,Y,Z:\Lambda)$ as described by eq. \eqref{eq: I_mermin} needed to explain a given degree of quantum violation for the Mermin inequality given by $\frac{\mathrm{M}-2}{4}$. A comparison between both curves shows that a higher degree of measurement dependence is required to explain the violation of the Mermin inequality with the same degree of quantum violation as the CHSH case, a point that can of experimental relevance when trying to close to violate measurement dependent Bell inequalities.
    }
    \label{fig: Mermin_vs_CHSH} 
\end{figure}

To compare the lower bound on the amount of measurement dependence required to explain a given degree of violation in the case of the Mermin inequality (Eq.~\eqref{eq: I_mermin}) versus the case of the the CHSH inequality (Eq.~~\eqref{eq:CHSHI}) is not straightforward as the Mermin inequality can be violated quantumly up to its algebraic maximum, while the CHSH inequality cannot.  We have therefore considered the degree of measurement dependence as a function of the ratio between the violation and the maximum quantum violation.  The result is plotted in Fig.~\ref{fig: Mermin_vs_CHSH} which demonstrates that Mermin requires more measurement dependence to explain away comparable violation ratios than CHSH.

Moreover, by assuming that some inputs never happen (as it is the case in the Mermin inequality), one obtains a lower bound for $I(X,Y,Z:\Lambda)$ that is exactly the same as in the CHSH scenario (see Appendix~\ref{app: Mermin}), indicating that it is possible to explore biased distribution of inputs in the analysis of measurement dependence.

\FloatBarrier

\section{Relating Causal Networks to Relaxations of Measurement Independence}\label{sec:causnetworks}

\begin{figure}[b]
    \centering
    \includegraphics[scale = 0.25]{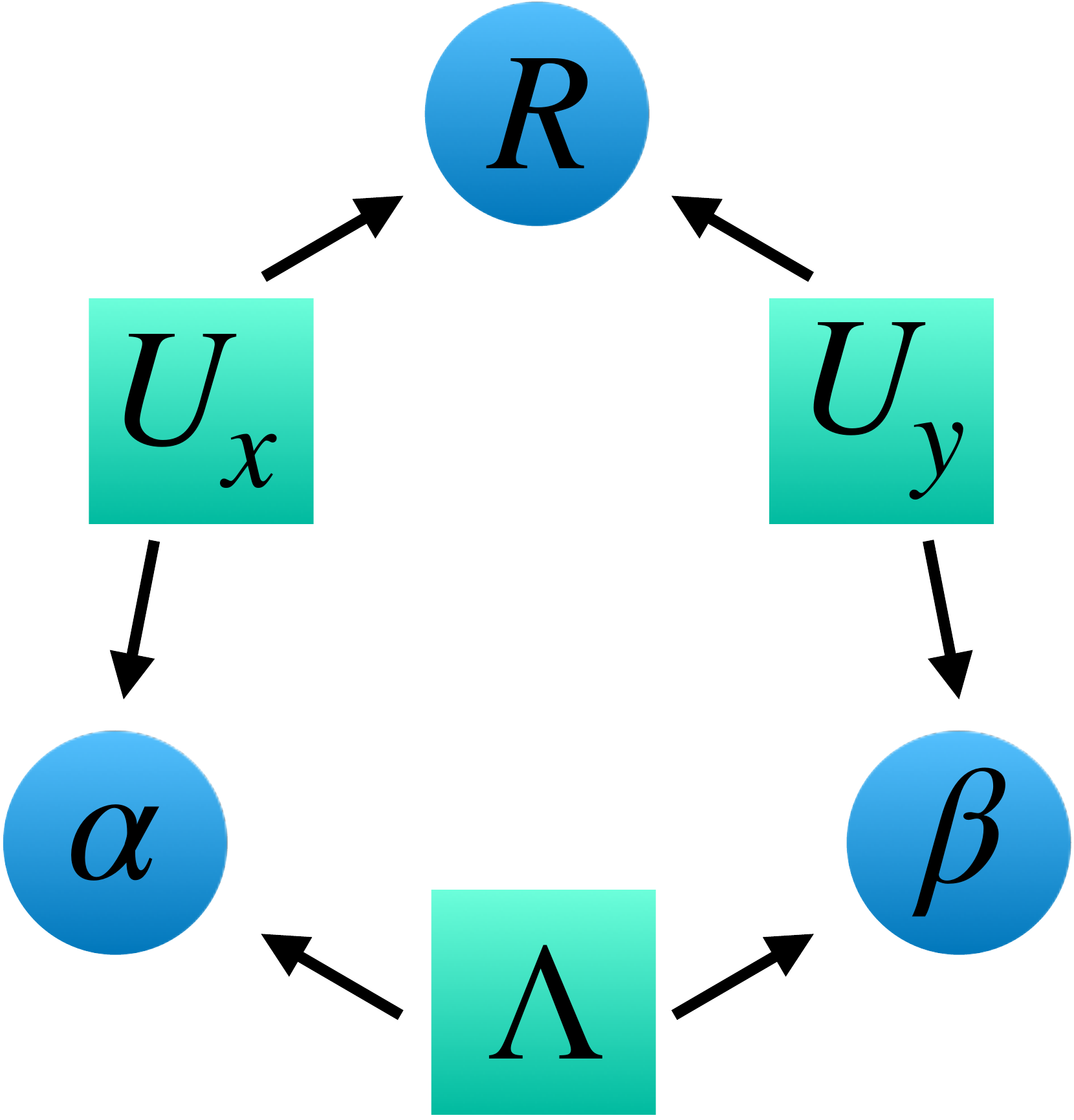}
    \caption{\textbf{Triangle network.} Pairwise independent sources generate the correlations between the three observable variables $\alpha$, $\beta$ and $R$.
    }
    \label{fig:Bell7}
\end{figure}

So far we have restricted our attention to the analysis of measurement dependence in standard Bell scenarios where the correlations between the distant parties is mediated by a single source. More recently, a number of new causal scenarios have started to be considered, typically consisting of many independent sources \cite{branciard2010characterizing,branciard2012bilocal,fritz2012beyond,fritz2016beyond,tavakoli2014nonlocal,chaves2016polynomial,poderini2020experimental,rosset2016nonlinear,andreoli2017maximal,renou2019genuine}. The paradigmatic example is the so-called triangle scenario, shown in Fig. \ref{fig:Bell7}. Its most prominent feature is that it can lead to nonlocal correlations even though it has no input variables \cite{fritz2012beyond,renou2019genuine}, an ingredient that until then was considered essential for the appearance of nonclassical behaviour.

In spite of the growing theoretical and experimental attention \cite{tavakoli2021bell}, progress in the analysis of nonclassical behaviour in such causal structures has been hampered by the fact that the set of correlations they define is nonconvex and very difficult to be characterized \cite{wolfe2019inflation,chaves2016polynomial}. In the following, we will show how our results for Bell scenarios with measurement dependence can be readily applied to derive new non-linear Bell inequalities for different classes of causal networks.

For the sake of example, we start by focusing on the triangle network. The most general correlations admissible in the triangle network have the form $p(\alpha,\beta,r) = $
\begin{align}\label{eq:triangle}
\smashoperator{\sum_{u_x,u_y,\lambda}}
p(\alpha\vert u_x,\lambda)p(\beta\vert u_y,\lambda)p(r\vert u_x,u_y)
p(u_x)p(u_y)p(\lambda)
\end{align}
which is precisely the same form as Eq.~\eqref{eq:triangleMD} under the association $\alpha\leftrightarrow (A,X)$ and $\beta\leftrightarrow (B,Y)$.

As shown in \cite{fritz2012beyond}, standard Bell scenario correlations can be mapped onto the triangle network, thus providing the first proof that such a causal structure can support nonclassical correlations. We generalize this result by noting that correlations in the \emph{nonstandard} Bell scenario with measurement dependence of Fig.~\ref{fig:Bell4} can be mapped \emph{bijectively} onto the triangle network via the common forms of Eqs.~\eqref{eq:triangleMD} and~\eqref{eq:triangle}. This two-way mapping allows us to translate results both ways between those scenarios.

Note that we have relaxed
the assumption that $X$ has a direct causal influence over $A$. Upon allowing for measurement dependence, there is no further loss of generality in treating $A$ and $X$ on an equal footing, i.e., as a single composite outcome variable $\alpha=(A,X)$ that is a function of the sources $U_x$ and $\Lambda$.\footnote{Formally, the latent projection of Fig.~\ref{fig:Bell3} is observationally equivalent to a different DAG wherein $A$ and $X$ are latent-common-cause connected and have identical causal ancestry, per Ref.~\cite{Evans2015}. This justifies merging them into a single composite variable without loss of generality.} We similarly merge $B$ and $Y$ into the single composite variable $\beta=(B,Y)$ . As detailed in the proof of Lemma~\ref{lem:generalub}, the bounds on the measurement dependence $I(X,Y:\Lambda)$ only assume such general dependence. Thus, all the results we have derived above can be directly applied to the triangle network.

\begin{lemma}
\label{lemma1}
Let $\mathcal{G}_{\text{Bell-MI+aux}}$ be the bipartite Bell scenario without the assumption of measurement independence supplemented with an auxiliary variable $R$ as per Fig.~\ref{fig:Bell4}. Let $\mathcal{G}_{\text{triangle}}$ be the causal scenario depicted in Fig.~\ref{fig:Bell7}. Then, a distribution $P(a,b,x,y,r)$ is incompatible with $\mathcal{G}_{\text{Bell-MI+aux}}$ if and only if  $P(\alpha=(a,x),\beta=(b,y),r)$ is incompatible with  $\mathcal{G}_{\text{triangle}}$.
\end{lemma}
\begin{cor} Any correlations compatible with the classical triangle network should respect the non-linear inequalities \eqref{eq:ineqCHSH1}, \eqref{eq:ineqCHSH2} and \eqref{eq:ineqCGLMP}. \end{cor}
\begin{cor}  For the special case of triangle scenario correlations where $H(X,Y\vert R)=0$, it follows that if $P(a,b\vert x,y)$ violates a traditional Bell inequality, then ${P(\alpha{=}(a,x),\beta{=}(b,y),r)}$ is incompatible with  $\mathcal{G}_{\text{triangle}}$.
\end{cor}

As a consequence of Lemma~\ref{lemma1} our results generalize the result of \citet{fritz2012beyond} in a number of ways, since Fritz's original argument given there was only applicable when ${H(X,Y\vert R)=0}$ and furthermore was unable to explicitly derive a testable Bell inequality. It is worth pointing out that different Bell inequalities able to witness quantum nonlocality in the triangle have already been derived \cite{PhysRevA.98.022113,renou2019genuine,vsupic2020quantum}. In particular, in \cite{PhysRevA.98.022113} a specific inequality has been obtained for testing the mapping between a standard Bell scenario and the triangle network as proposed in \cite{fritz2012beyond}. Similarly, our inequalities \eqref{eq:ineqCHSH1} and \eqref{eq:ineqCHSH2} will also witness the nonclassical behaviour whenever the CHSH inequality is violated (if $H(X,Y\vert R)=0$).

Moreover, our construction can be extended to prove the possibility of nonclassical behaviour in causal networks of growing size and where the variables can assume different cardinalities. As a generalization of the triangle causal structure we will consider any causal network inspired by the so-called bipartite graphs \cite{aaberg2020semidefinite,PhysRevLett.125.110505} and composed of two layers: a first layer corresponding to $o+s$ sources labelled as $\left\{ \Lambda_1,\dots,\Lambda_o, U_1,\dots,U_s \right\}$ acting as common causes to the $n+m$ observable variables $\alpha_1,\dots,\alpha_n,R_1,\dots,R_m$. In the following, we will restrict our attention to two particular classes of these general networks and in particular derive nonlinear Bell inequalities which can be violated by quantum correlations in those networks. In the first class we fix $o=1$ and let $s=n$, whereas in the second class we fix $s=2$ and let $o=n-2$.

\subsection{``2's \& $n$" networks} 
As a first case, we consider networks of $n+1$ parties in which all parties except for $R$ are connected by an $n$-way source, whereas $R$ is connected to every other individual party by a $2$-way source. An example of this if given in Fig.~\ref{fig:Bell8replacement}. We call such a scenario a ``2's \& $n$" network. As a second case we will consider the cyclic network of degree $n$ where $n-2$ sources $\Lambda_i$ connect pairs of $\alpha_i$'s (with $i=1,\dots,n$), source $U_1$ connects $\alpha_1$ to $R$ and $U_2$ connects $\alpha_n$ to $R$ (see Fig. \ref{fig:Bell9}). 
\begin{figure}[b]
    \centering
    \includegraphics[scale = 0.3]{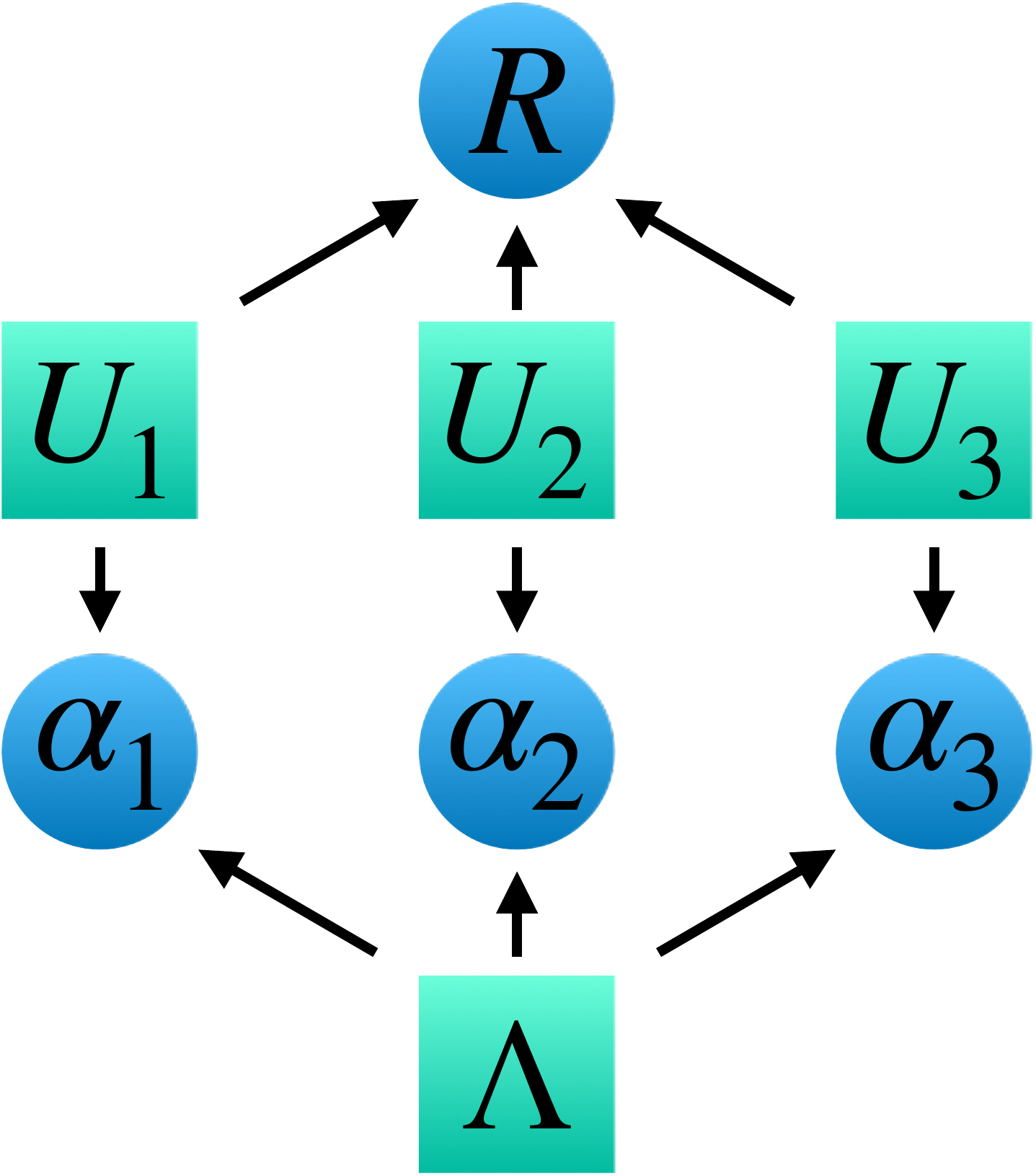}
    \caption{\textbf{``2's \& $\boldsymbol{n}$" network} for $n=3$. One observable source connects three of the four observable variables; the other sources only connect pairs, however. More generally, a ``2's \& $n$" consists of observable variables ${\alpha_1,...,\alpha_n}$ in addition to $R$. Notice that the triangle network (see Fig. \ref{fig:Bell7}) is a ``2's \& $n$" network with $n=2$. 
    }
    \label{fig:Bell8replacement}
\end{figure}
A ``2's \& $n$" network admits distributions of the form $p(\alpha_1,...,\alpha_n,r) =$
\begin{align}\label{eq:2sn}
 \smashoperator{\sum_{u_1,...,u_n,\lambda}}
p(r\vert u_1,...,u_n)p(\lambda)\prod_{i=1}^n p(\alpha_i\vert u_i,\lambda)p(u_i).
\end{align}
The mapping between ``2's \& $n$" networks and multipartite Bell scenarios with measurement dependence and an auxialliary variable $R$ is ready evident by comparing the forms of Eqs.~\eqref{eq:multiBellform} and \eqref{eq:2sn}. The are equivalent under the relabelling $\alpha_i\leftrightarrow (A_i,X_i)$.

\begin{lemma}
\label{lem:MultiBell}
Let $\mathcal{G}_{\text{MultiBell-MI+aux}}$ be the multipartite Bell scenario without the assumption of measurement independence supplemented with an auxiliary variable $R$ as per Fig.~\ref{fig:Bell6}. Let $\mathcal{G}_{\text{``}2\text{'s \& }n\text{"}}$ be a causal scenario of the family depicted in Fig.~\ref{fig:Bell8replacement}. Then, a distribution ${P(a_1,...,a_n,x_1,...x_1,r)}$ is incompatible with $\mathcal{G}_{\text{MultiBell-MI+aux}}$ if and only if  ${P(\alpha_1=(a_1,x_1),...,\alpha_n=(a_n,x_n),r)}$ is incompatible with $\mathcal{G}_{\text{``}2\text{'s \& }n\text{"}}$.
\end{lemma}
\begin{cor} Any correlations compatible with the ``2's \& $n$" network for $n=3$ depicted in Fig.~\ref{fig:Bell8replacement} should respect the non-linear inequality~\eqref{eq:multiMDexample}. \end{cor}
\begin{cor}  For the special case of ``2's \& $n$" scenario correlations where $H(X_1,...,X_n,\vert R)=0$, it follows that if ${P(a_1,...,a_n\vert x_1,...,x_n)}$ violates a traditional multipartite Bell inequality, then ${P(\alpha_1=(a_1,x_1),...,\alpha_n=(a_n,x_n),r)}$ is incompatible with $\mathcal{G}_{\text{``}2\text{'s \& }n\text{"}}$.
\end{cor}

Firstly, note that the generic multipartite bounds implied by combining inequality~\eqref{eq:generic} with Lemma~\ref{lem:generalub} remain valid for ``2's \& $n$" networks, thus providing a general non-linear Bell inequality of the form~\eqref{eq:generalbi}  for it. To see that, notice that in the proof of Lemma~\ref{lem:generalub}, the crucial step~\eqref{eq:mincausalassumption} which invokes the causal structure under analysis only makes use of the causal assumption that $R$ is independent of $\Lambda$, a condition fulfilled by ``2's \& $n$" networks. 
As a consequence, the generic non-linear Bell inequality~\eqref{eq:generic} will hold, where $I$ is a function of the conditional probability distribution ${p(a_1,\dots,a_n \vert x_1,\dots,x_n)}$.

\subsection{Cyclic networks}
Next we will consider the cyclic network. By mapping it onto an $n$-locality scenario \cite{branciard2010characterizing,branciard2012bilocal,mukherjee2015correlations}, we will be able to solve an open problem in the characterization of quantum correlations in networks. More specifically, although  it has been proven in \cite{fritz2012beyond} that the cyclic scenario of Fig. \ref{fig:Bell9} gives rise to nonclassical correlations, the proof there relies of nonclassicality of a post-quantum nature~\cite{popescu1994quantum}. It was left open whether a nonclassical behaviour that is quantumly realizable would be possible. 
The basic idea will be to map the bilocality scenario shown in Fig. \ref{fig:Bell10} onto this cyclic scenario with $n=4$. Notice, that in this particular bilocality scenario, there are no external inputs acting as causal parents of the central node $A_2$.
Thus we achieve a mapping by setting $\alpha_1{=}(A_1,X_1)$, $\alpha_2{=}A_2$ and $\alpha_3{=}(A_3,A_3)$. 

\begin{figure}[t]
    \centering
    \includegraphics[scale = 0.18]{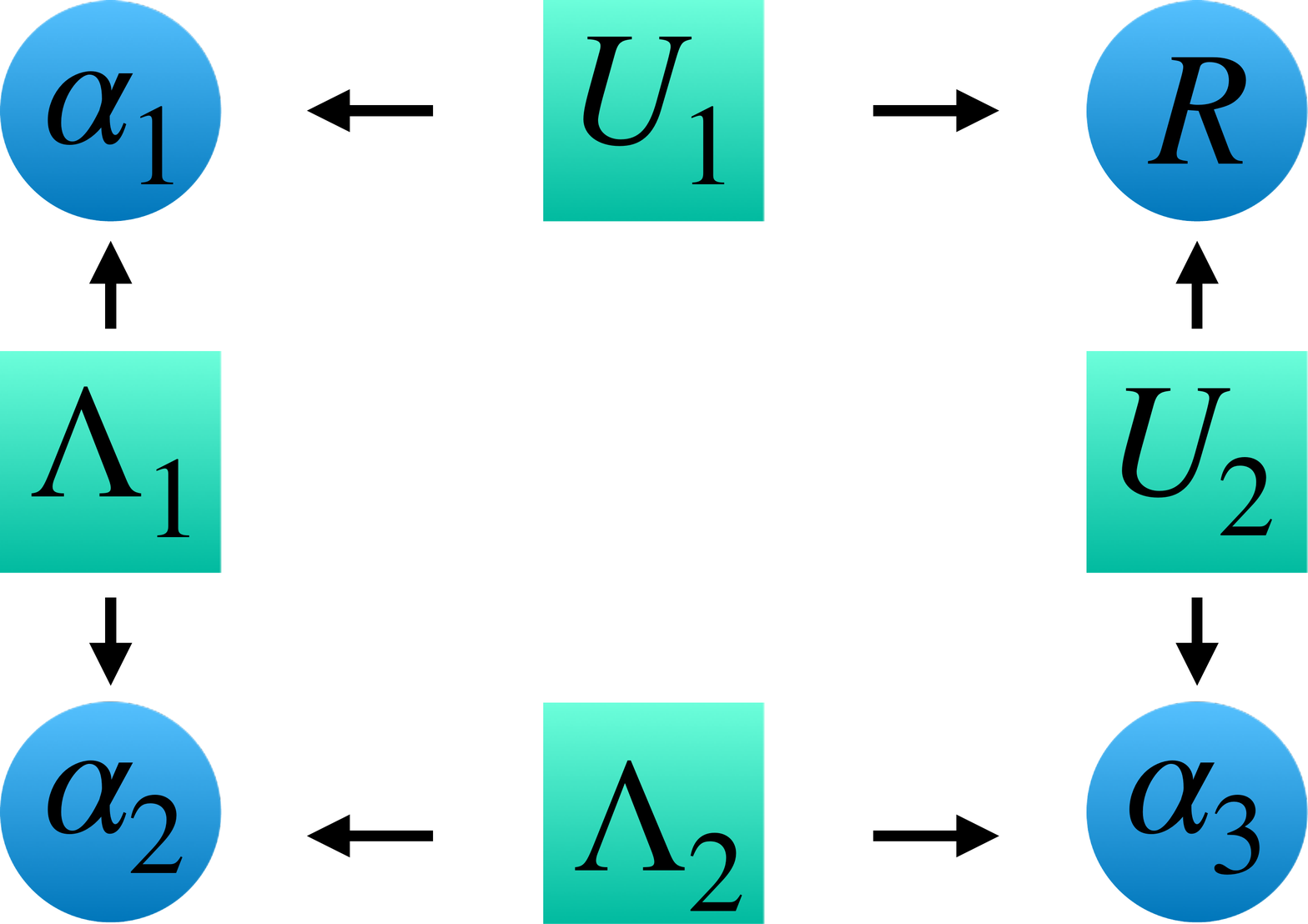}
    \caption{\textbf{Cyclic causal structure with $n=4$.} Each observable variable its connected to its neighbour via a common source. Notice that the triangle network (see Fig. \ref{fig:Bell7}) is a particular case of cyclic network with $n=3$. We index the observable variables in and $n$-cyclic scenario as $\alpha_1,...,\alpha_{n-1}$ along with $R$.}
    \label{fig:Bell9}
\end{figure}

There are many nonlinear Bell inequalities $I \leq L$ (where $I$ is a polynomial function of $p(a_1,a_2,a_3 \vert x_1,x_3)$) derived to characterize the causal structure of the bilocality scenario and that can be violated by quantum correlations. As an example, we have the bilocality inequality \cite{branciard2010characterizing,branciard2012bilocal}:
\begin{align}
&\sqrt{I}+\sqrt{J} \leq 2,\quad\text{where}
\\\nonumber I&\coloneqq \sum_{x_1,x_3} \langle a^{x_1}_1a_2a_3^{x_3}\rangle,\quad\text{and}
\\\nonumber J&\coloneqq\sum_{x_1,x_3} (-1)^{x_1+x_3} \langle a^{x_1}_1a_2a_3^{x_3} \rangle,\quad\text{and}
\\\nonumber \langle a^{x_1}_1a_2a_3^{x_3} \rangle&\coloneqq\sum_{a_1,a_2,a_3} (-1)^{a_1+a_2+a_3}p(a_1,a_2,a_3 \vert x_1,x_3)
\end{align}
and where all input and output variables assume the values $0$ or $1$.

\begin{lemma}
Let $\mathcal{G}_{\text{$n$-locality-MI+aux}}$ be the $n$-locality scenario without the assumption of measurement independence supplemented with an auxiliary variable $R$ as per Fig.~\ref{fig:Bell10} and generalizations thereof. Let $\mathcal{G}_{(n+2)-\text{cyclic}}$ be the causal scenario depicted in Fig.~\ref{fig:Bell9} and generalizations therefore. Then, a distribution ${P(a_1,a_2,...a_n,a_{n+1},x_1,x_{n+1},r)}$ is incompatible with  $\mathcal{G}_{n\text{-locality-MI+aux}}$ if and only if ${P(\alpha_1{=}(a_1,x_1), \alpha_2{=}a_2, ...,\alpha_{n+1}{=}(a_{n+1},x_{n+1}),r)}$ is incompatible with $\mathcal{G}_{(n+2)-\text{cyclic}}$.
\end{lemma}
\begin{cor}
For the special case of $n$-cyclic scenario correlations where ${H(X_1,X_{n-1}\vert R)=0}$, it follows that if ${{P(a_1,a_2,...,a_{n-1}\vert x_1,x_{n-1})}}$ violates a traditional $n$-locality inequality, then ${P(\alpha_1{=}(a_1,x_1), \alpha_2{=}a_2,...,\alpha_{n-1}{=}(a_{n+1},x_{n-1}),r)}$ is incompatible with $\mathcal{G}_{\text{cyclic}}$.
\end{cor}

\begin{figure}[b]
    \centering
    \includegraphics[scale = 0.34]{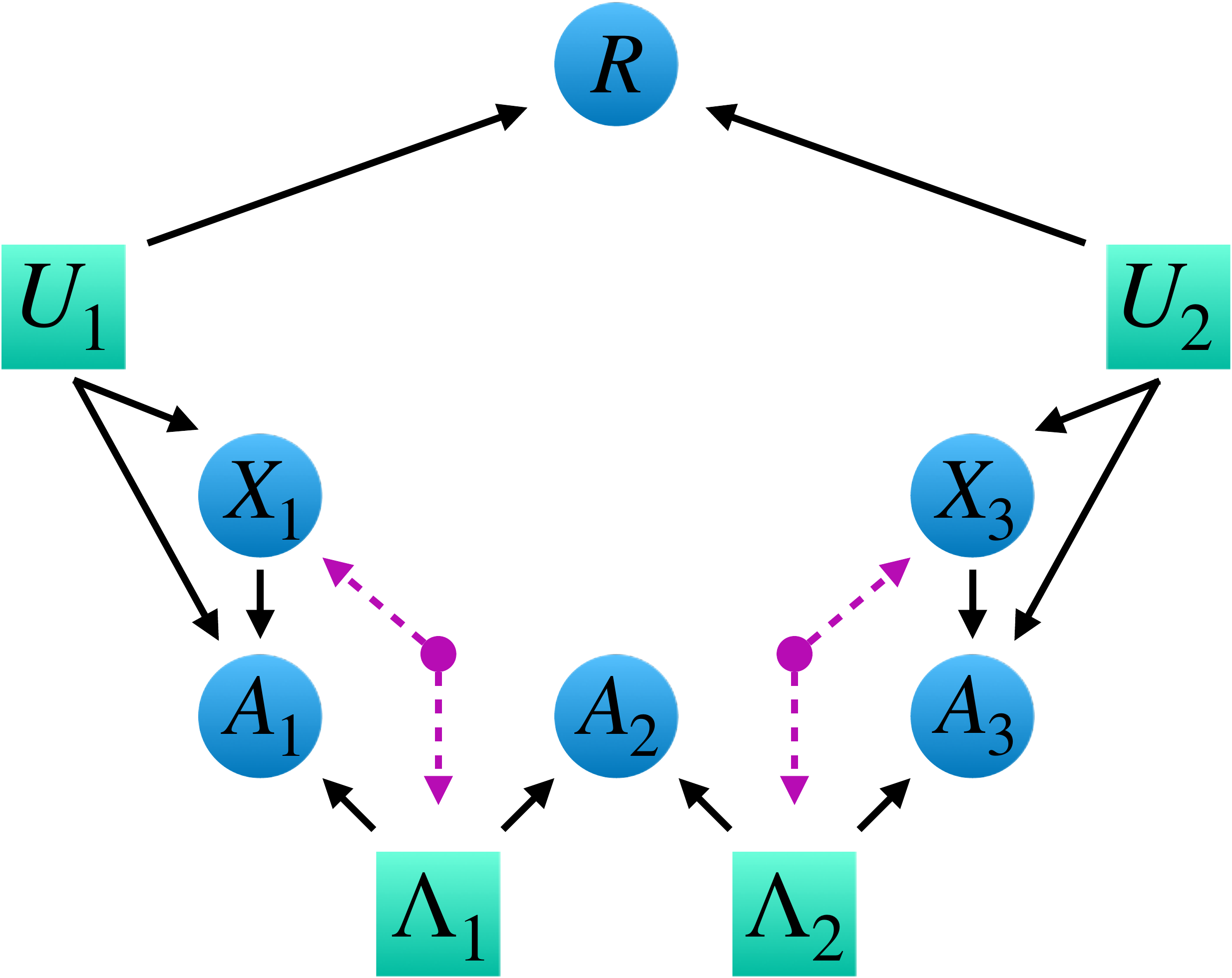}
    \caption{\textbf{Bilocality causal structure,} shown here with measurement dependence (sources and local inputs potentially correlated) along with an auxiliary variable $R$. In the absence of measurement dependence this scenario is akin to an entanglement swapping scenario~\cite{zukowski1993event} with a central node sharing correlations with two peripheral nodes via two independent sources. More generally, in an $n$-locality scenario with measurement dependence and an auxiliary variable, we have a chain of outcome variables $A_1,...,A_{n+1}$ along with $R$, and only the two peripheral nodes $A_1$ and $A_{n+1}$ have inputs.
  }
    \label{fig:Bell10}
\end{figure}

Following the recipe in \cite{chaves2015unifying}, one can also relate the violation of such inequalities with the degree of measurement dependence ${\mathcal{M}\coloneqq\sum_{\lambda_1,\lambda_2} \vert p(x_1,x_3,\lambda_1,\lambda_2)- p(x_1,x_3)p(\lambda_1)p(\lambda_2)\vert}$  required to explain it. As before, one can then generally write $g(I)\leq \mathcal{M}$ where now $g$ is a polynomial function of the Bell inequality $I$ such that $g(I) >0$ if the inequality is violated and $g(I) \leq 0$ otherwise. Thus, we also obtain
\begin{equation}
 g(I)  \leq \sqrt{\frac{H(X_1,X_3\vert R)}{\log_2 e}}.
\end{equation}
In particular, if $X_1$ and $X_3$ are perfectly correlated with and $R$, then $H(X_1,X_3\vert R)=0$ and the quantum violation of any bilocality inequality would be enough to witness quantum nonclassicality in such a cyclic network.

Leveraging the linear chain $n$-locality scenario considered in \cite{mukherjee2015correlations}, one can demonstrate nonclassicality in cyclic networks of arbitrary size. In the linear chain $n$-locality scenario,  we have ${n{+}1}$ observable variables $A_i$ of which only $A_1$ and $A_{n+1}$ have inputs, denoted $X_1$ and $X_{n+1}$ respectively. If we treat $\alpha_1=(a_1,x_1)$ and $\alpha_N=(a_{n+1},x_{n+1})$ and otherwise $\alpha_i = a_i$, and we close the cycle by introducing the sources $U_1$ and $U_2$ that connect $\alpha_1$ and $\alpha_n$ to the variable $R$, we obtain the ${(n{+}2)}$-cycle causal network.
The joint distribution on observed variables that are compatible with this network are
\begin{widetext}
\begin{align}
p(\alpha_1,...\alpha_{n+1},r)=p(r\vert u_1,u_2)p(\alpha_1\vert u_1 \lambda_1)p(\alpha_{n+1}\vert u_2,\lambda_n)p(\lambda_1)\prod_{i=2}^n p(\alpha_i\vert \lambda_{i-1},\lambda_i)p(\lambda_n).
\end{align}
\end{widetext}
Compatible distributions in the  $n$-locality scenario with measurement dependence and an auxiliary variable $R$, that is, generalizations of Fig.~\ref{fig:Bell10}, are of \emph{precisely the same form} but with the relabelling 
\begin{align}
\alpha_i \leftrightarrow \begin{cases}
(A_i,X_i) &\text{if }\ i=1\ \text{ or }\ i=n+1,\\
A_i &\text{if }\ 1<i<n+1.
\end{cases}
\end{align}

As before, if $X_1$ and $X_{n+1}$ are perfectly correlated with $R$ then $H(X_1,X_{n+1}\vert R)=0$ and the violation of any $n$-locality inequality bounding the linear chain scenario will suffice to demonstrate nonlocality also in the cyclic scenario. An example of such inequality was provided in \cite{mukherjee2015correlations}.

\section{Discussion}\label{sec:conclusions}

If observed statistical correlations are found to violate a Bell inequality, then one can conclude that these correlations cannot be explained by a classical causal model having the causal structure of Fig.~\ref{fig:Bell}. The moniker of `Bell nonlocality' has traditionally been assigned to this phenomenon because one can avoid the contradiction through a modification of the causal structure of Fig.~\ref{fig:Bell} that incorporates a (nonlocal) causal influence from the setting on one side to the outcome on the other. It is well known, however, that there are other opportunities for evading the contradiction.  In particular, one can modify the causal structure of Fig.~\ref{fig:Bell} by
relaxing the assumption that there is no causal influence from $\Lambda$ to the setting variables $X$ or $Y$ and no common cause of $\Lambda$ and $X$ or of $\Lambda$ and $Y$, that is, by relaxing the assumption of measurement independence. 
We have argued that there are two distinct causal mechanisms by which the assumption of measurement independence might fail in a Bell experiment. The first mechanism is a violation of independence between $\Lambda$ and variables that causally determine the settings $X$ and $Y$.  If the settings are made to depend causally on cosmic photons~\cite{handsteiner2017cosmic,rauch2018cosmic}, 
then this version of the assumption seems especially plausible, since denying it seems to require assuming a superdeterministic world wherein everything is potentially correlated with everything else. The second mechanism is one wherein some systems that {\em mediate} the influence of the ultimate causal determinants of the settings (such as cosmic photons) become influenced by $\Lambda$ or by a variable that also influences $\Lambda$.  In other words, even though the variables that determine the settings may start out independent of $\Lambda$, they might {\em  become} correlated with it (e.g., when the cosmic photons enter the laboratory).  
It is this second class of mechanism for measurement dependence that we have here shown can be subjected to an experimental test.

We embed the Bell causal structure in a larger network and assume the independence of the variables that causally determine the settings (i.e., we rule out, by assumption, the first mechanism for achieving measurement dependence), and show that in this case it is possible to upper bound the amount of measurement dependence (of the second kind) based on observational data. Combining these upper bounds with previous lower bounds for the amount of measurement dependence required to explain a given violation of Bell inequalities, we are able to derive nonlinear Bell-type inequalities whose violation is a proof of 
can witness nonclassicality in spite of the presence of some measurement dependence (of the second kind).

To our knowledge, this is the first demonstration of the possibility of putting an upper bound on the amount of measurement dependence in a Bell experiment, a feature that deserves further theoretical and experimental investigation. It is worth remarking that the results in \cite{vsupic2020quantum} can be understood as complementary to ours. There, focusing on the triangle network, it has been shown that under the assumption of perfect correlation between some variables (ruling out any possible measurement dependence of the second kind) they can witness nonclassicality even allowing correlations between the sources (thus allowing measurement dependence of the first kind, even though the latter cannot be upper bounded by observational data).

Following that, we have also shown how measurement-dependent Bell causal structures can be readily mapped onto causal networks of growing size and complexity, a field of research that is attracting growing attention but for which advances have been hampered by the difficulty in deriving Bell inequalities. By doing so,  we have derived a robust Bell inequality especially suited to test the Fritz correlations \cite{fritz2012beyond} in the triangle network and which, in contrast to previous attempts, does not require perfect correlations between some subset of variables~\cite{vsupic2020quantum}, which is a welcome feature for achieving an experimental implementation. 
Finally, by mapping fully connected networks onto multipartite Bell scenarios and by mapping cyclic networks to the linear $n$-locality scenario, we were able to show that such networks can give rise to correlations that witness nonclassicality.

We believe that our results are just a first step towards a better understand of measurement dependent causal models and how these can be tested experimentally. For instance, upper bounds like the one in \eqref{eq:generalbi} employ the Pinsker inequality, which is known to be non-tight. Can we employ more modern techniques such as those in \cite{hall2020measurement}, the covariance \cite{aaberg2020semidefinite,pozas2019bounding} or the inflation technique \cite{wolfe2019inflation}, in order to provide better lower and upper bounds to $\mathcal{M}$ and $I(X,Y:\Lambda)$ and thus improve our inequalities? We have proven that is possible for the tripartite scenario by deriving new results for the Mermin inequality.

It is noteworthy that the analysis we have carried out  here can also be extended to analyze measurement dependence in different causal structures, in particular in the instrumental scenario (which plays a central role in the field of causal inference). As is further explored in \cite{morenoetal}, measurement independence is also a crucial assumption in causal inference, the violation of which has important consequences to the analysis of cause and effect in empirical data \cite{kedagni2020generalized,balke1997bounds}.

Finally, we believe that the techniques we have introduced here can offer a way to tackle open questions in the study of networks. As mentioned before, the first example of nonclassical correlations in the triangle network was provided by a mapping of that network onto the usual Bell network \cite{fritz2012beyond}. Other examples of nonclassicality that do not hinge directly on Bell's theorem are known \cite{renou2019genuine}. However, it remains an open question of how can one prove that the nonclassicality observed in a given network is truly different from that in Bell's theorem. Further exploration of the connection between Bell scenarios with measurement dependence and networks may offer a way to better understand their similarities and their differences. 
We hope that our results might trigger future research along all of these directions.

\section{Acknowledgements} 
This work was supported by The John Templeton Foundation via the grant Q-CAUSAL No 61084 and via The Quantum Information Structure of Spacetime (QISS) Project (qiss.fr) (the opinions expressed in this publication are those of the author(s) and do not necessarily reflect the views of the John Templeton Foundation)  Grant Agreement No.  61466, by the Serrapilheira Institute (Grant No. Serra-1708-15763), the Brazilian National Council for Scientific and Technological Development (CNPq) via the National Institute for Science and Technology on Quantum Information (INCT-IQ) and Grants No. 307172/2017-1 and No. 311375/2020-0, the Brazilian agencies MCTIC and MEC , by MIUR via PRIN 2017 (Progetto di Ricerca di Interesse Nazionale): project QUSHIP (2017SRNBRK). This research was supported by Perimeter Institute for Theoretical Physics. Research at Perimeter Institute is supported in part by the Government of Canada through the Department of Innovation, Science and Economic Development Canada and by the Province of Ontario through the Ministry of Colleges and Universities. The opinions expressed in this publication are those of the authors and do not necessarily reflect the views of the supporting institutions.

\bibliographystyle{apsrev4-2-wolfe}
\nocite{apsrev42Control}
\bibliography{mrefs}\clearpage

\onecolumngrid
\appendix
\section{Mutual information bound for Mermin inequality}

In this appendix, we follow similar steps to those introduced in reference \cite{hall2020measurement} and we derive a tight informational lower bound of the measurement dependence demanded for a given violation of Mermin's inequality. To that aim, we first show how any violation of this inequality imposes bounds on the distribution $p(x,y,z|\lambda)$, then, considering those bounds, we minimize the mutual information $I(X,Y,Z:\Lambda)$.

\label{app: Mermin}
In this scenario, the probability $p(a,b,c|x,y,z)$, for $a, b, c, x, y, z \in \{0,1\}$, can be factorized as
\begin{eqnarray}
\label{eq: prob factorization}
p(a,b,c|x,y,z) & = &
\sum_{\lambda\in\Lambda}p(\lambda|x,y,z)\delta_{a,f_{\lambda}(x)}\delta_{b,f_{\lambda}(y)}\delta_{c,f_{\lambda}(z)}
\end{eqnarray}

In which case the full correlators are given by,
\begin{eqnarray}
\label{eq: full correlator}
\langle A_x B_y C_z \rangle & = & \sum_{\lambda}A_{x}(\lambda)B_y(\lambda)C_z(\lambda)p(\lambda|x,y,z),
\end{eqnarray}
for $A_x,B_y,C_z\in\{-1,1\}$, given by
\begin{eqnarray}
\label{eq: A_x, B_y and C_z}
\nonumber
A_x(\lambda) & = & \sum_{a}(-1)^{a}\delta_{a,f_{\lambda}(x)}\\
\nonumber
B_y(\lambda) & = & \sum_{b}(-1)^{b}\delta_{b,f_{\lambda}(y)}\\
C_z(\lambda) & = & \sum_{c}(-1)^{c}\delta_{c,f_{\lambda}(z)}
\end{eqnarray}

Recall that Mermin's inequality \cite{mermin1990,ardehali1992bell,belinskiui1993}  reads
\begin{eqnarray}
\nonumber
M = \sum_{\left(x,y,z\right) \in O}(-1)^{xyz}\langle A_xB_yC_z\rangle \leq 2,
\end{eqnarray}
in which $O$ stand for the set of odd total parity outcomes, i.e. $O=\{(1,0,0),\ (0,1,0),\ (0,0,1),\ (1,1,1)\}$.

Substituting via \eqref{eq: full correlator}, applying Bayes' Rule such that $p(\lambda|xyz)=p(x,y,z|\lambda)p(\lambda)/p(x,y,z)$, and setting $p(x,y,z) = \frac{1}{8}$ we obtain,
\begin{eqnarray}
\nonumber
M & = & 8\sum_{\lambda}\sum_{\left(x,y,z\right)\in O}(-1)^{xyz}A_{x}(\lambda)B_y(\lambda)C_z(\lambda)p(x,y,z|\lambda)p(\lambda).
\end{eqnarray}
By defining the sets $\mathcal{L}_{\mu,\eta,\nu}$, for $\mu,\eta,\nu\in\{0,1\}$, as
\begin{eqnarray}
\label{eq: partition}
\mathcal{L}_{\mu,\eta,\nu} = \left\{\lambda\in\Lambda |  A_1(\lambda) = (-1)^{\mu}A_0(\lambda),\  B_1(\lambda)  = (-1)^{\eta}B_0(\lambda),\ C_1(\lambda) = (-1)^{\nu}C_0(\lambda)\right\},
\end{eqnarray}
we can rewrite $M$ as follows,
\begin{eqnarray}
\label{eq: M temp}
M & = & 8\sum_{\mu,\eta,\nu}\sum_{\lambda\in\mathcal{L}_{\mu,\eta,\nu}}A_0(\lambda)B_0(\lambda)C_0(\lambda)p(\lambda)\Gamma_{\mu,\eta,\nu}
\end{eqnarray}
in which
\begin{eqnarray}
\nonumber
\Gamma_{\mu,\eta,\nu} & = & \sum_{\left(x+y+z\right)\in O}(-1)^{xyz + \mu x + \eta y + \nu z}p(x,y,z|\lambda)
\end{eqnarray}
Recognizing that $p(O|\lambda) = p(0,0,1|\lambda) + p(0,1,0|\lambda) + p(1,0,0|\lambda) + p(1,1,1|\lambda)$, we can write
\begin{eqnarray}
\label{eq: Gamma}
\Gamma_{\mu,\eta,\nu}  =  (-1)^{\mu\eta + \mu\nu + \eta\nu}\left[p(O|\lambda) - 2p(\eta\oplus\nu\oplus 1,\mu\oplus\nu\oplus 1,\mu\oplus\eta\oplus 1|\lambda)\right],
\end{eqnarray}
in which $\oplus$ represents sum mod 2.

Using the above result in equation \ref{eq: M temp},
\begin{eqnarray}
\nonumber
M & = & \sum_{\mu,\eta,\nu}\sum_{\lambda\in\mathcal{L}_{\mu,\eta,\nu}}(-1)^{\mu\eta + \mu\nu + \eta\nu}A_0(\lambda)B_0(\lambda)C_0(\lambda)p(\lambda)p(O|\lambda)\left[1 - 2p(x=\eta\oplus\nu\oplus 1,y = \mu\oplus\nu\oplus 1,z = \mu\oplus\eta\oplus 1|\lambda,O)\right].
\end{eqnarray}

By choosing $A_0(\lambda)B_0(\lambda)C_0(\lambda) = \mbox{sgn}\left((-1)^{\mu\eta + \mu\nu + \eta\nu}\left[1 - 2 p(x=\eta\oplus\nu\oplus 1,y = \mu\oplus\nu\oplus 1,z = \mu\oplus\eta\oplus 1|\lambda,O)\right]\right)$, for $\lambda\in\mathcal{L}_{\mu,\eta,\nu}$, we get
\begin{eqnarray}
\nonumber
M & \leq & 8\sum_{\mu,\eta,\nu}\sum_{\lambda\in\mathcal{L}_{\mu,\eta,\nu}}p(\lambda)p(O|\lambda)\left|1 - 2 p(x=\eta\oplus\nu\oplus 1,y = \mu\oplus\nu\oplus 1,z = \mu\oplus\eta\oplus 1|\lambda,O)\right|
\end{eqnarray}

For $(x,y,z)\in O$, let $p_{min}\leq\frac{1}{4}$ be the infimun of $p(x,y,z|\lambda,O)$, then,
\begin{eqnarray}
\nonumber
M & \leq & 8\sum_{\mu,\eta,\nu}\sum_{\lambda\in\mathcal{L}_{\mu,\eta,\nu}}p(\lambda)p(O|\lambda)\left(1 - 2p_{min}\right)\\
\nonumber
& \leq & 8\left(1 - 2p_{min}\right)p(O)\\
\nonumber
& \leq & 4 - 8p_{min},
\end{eqnarray}
or yet,
\begin{eqnarray}
\label{eq: bound p_min}
p_{min}\leq\frac{4 - M}{8}
\end{eqnarray}
saturation is achieved for,
\begin{eqnarray}
p(x=\eta\oplus\nu\oplus 1,y = \mu\oplus\nu\oplus 1,z = \mu\oplus\eta\oplus 1|\lambda,O) = p_{min}\quad\mbox{or}\quad 1 - p_{min}
\end{eqnarray}

The mutual information reads,
\begin{eqnarray}
\nonumber
I(X,Y,Z:\Lambda) & = & H(X,Y,Z) - \sum_{\mu,\eta,\nu}\sum_{\lambda\in\mathcal{L}_{\mu,\eta,\nu}}p(\lambda)H_{\lambda}(X,Y,Z)\\
& = & 3 - \sum_{\mu,\eta,\nu}\sum_{\lambda\in\mathcal{L}_{\mu,\eta,\nu}}p(\lambda)H_{\lambda}(X,Y,Z),
\end{eqnarray}
and is minimized by making $H_{\lambda}(X,Y,Z)$ as large as possible, which is achieved with distributions $p(x,y,z|\lambda)$ as close to a uniform distribution as possible. Thus,
\begin{eqnarray}
p(x,y,z|\lambda) = p_{min} = \frac{4 - M}{8}.
\end{eqnarray}
Notice above that we choose $p_{min}$ insted of $1-p_{min}$, otherwise we would end up with a distribution further from the uniform one.

For $\lambda\in\mathcal{L}_{\mu,\eta,\nu}$, this leads to distributions of the form:
\begin{eqnarray}
p(x,y,z|\lambda) = \left\{\begin{array}{lll}
      \frac{p_{min}}{2} &,\quad\mbox{for }(x,y,z)= (\eta\oplus\nu\oplus 1,\mu\oplus\nu\oplus 1, \mu\oplus\eta\oplus 1)\\
      \frac{1 - p_{min}}{6}&,\quad\mbox{for }(x,y,z)\in O\mbox{ and }(x,y,z)\neq (\eta\oplus\nu\oplus 1,\mu\oplus\nu\oplus 1, \mu\oplus\eta\oplus 1) \\
      \frac{1}{8} &, \quad \mbox{for }(x,y,z)\not\in O
    \end{array}\right.
\end{eqnarray}
which can be used for each case of $\mu,\eta,\nu$ always returning the same value for $H_{\lambda}(X,Y,Z)$, so the mutual information reads,
\begin{eqnarray}
\nonumber
I(X,Y,Z:\Lambda) & \geq & 3 + \frac{p_0}{2}\log p_0 - \frac{p_0}{2}\log 2 + \frac{(1-p_0)}{2}\log(1-p_0) - \frac{(1-p_0)}{2}\log 6 -\frac{1}{2}\log8\\
\nonumber
& = & 1 + \frac{1}{2}h(p_0) - \frac{(1-p_0)}{2}\log 3\\
\nonumber
& = & 1 + \frac{1}{2}h\left(\frac{4-M}{8}\right) - \frac{4+M}{16}\log 3
\end{eqnarray}

If now we set $p(x,y,z) = \frac{1}{4}$ if $(x+y+z)\in O$, then we get
\begin{eqnarray}
M & = & 4\sum_{\mu,\eta,\nu}\sum_{\lambda\in\mathcal{L}_{\mu,\eta,\nu}}A_0(\lambda)B_0(\lambda)C_0(\lambda)p(\lambda)\Gamma'_{\mu,\eta,\nu}
\end{eqnarray}
in which
\begin{eqnarray}
\Gamma'_{\mu,\eta,\nu}  =  (-1)^{\mu\eta + \mu\nu + \eta\nu}\left[1 - 2p(\eta\oplus\nu\oplus 1,\mu\oplus\nu\oplus 1,\mu\oplus\eta\oplus 1|\lambda)\right].
\end{eqnarray}

While this new constraint leads to the same bound on $p_{min}$, shown in equation \ref{eq: bound p_min}, in this case we must adapt the distribution maximizing $H_{\lambda}(x,y,z)$ to the following,
\begin{eqnarray}
p(x,y,z|\lambda) = \left\{\begin{array}{lll}
     p_{min} &,\quad\mbox{for }x=\eta\oplus\nu\oplus 1,y = \mu\oplus\nu\oplus 1,z = \mu\oplus\eta\oplus 1\\
     \frac{1 - p_{min}}{3}&,\quad\mbox{for }x\neq\eta\oplus\nu\oplus 1,\mbox{ or } y \neq \mu\oplus\nu\oplus 1,\mbox{ or }z \neq \mu\oplus\eta\oplus 1,\;\mbox{and }x+y+z\in O \\
     0 &, \quad \mbox{for }x+y+z\not\in O
\end{array}\right.
\end{eqnarray}
for $\lambda\in\mathcal{L}_{\mu,\eta,\nu}$.

This leads to,
\begin{eqnarray}
\nonumber
I(X,Y,Z:\Lambda) & \geq & 2 + p_0\log p_0 + (1-p_0)\log(1-p_0) - (1-p_0)\log 3\\
\nonumber
& \geq & 2 - h\left(\frac{4-M}{8}\right) - \left(\frac{4+M}{8}\right)\log 3
\end{eqnarray}

\section{Nosignaling condition}
\label{app: nosig}

As in reference \cite{hall2020measurement}, we also address the question of imposing the extra constrain of nosignaling. The answer is exactly the same: the existence of a signaling distribution leading to an arbitrary value of $M$ implies the existence of a nosignaling distribution featuring the same mutual information between the inputs and the hidden variable and value of $M$.

\begin{proof}
Assume that exists a distribution $p(a,b,c|x,y,z)$ as in \ref{eq: prob factorization}, which does not satisfy the nosignaling condition and for which $M(p(a,b,c|x,y,z))=M^*$. Now, let us build a distribution $\tilde{p}(a,b,c|x,y,z)$ as follows,
\begin{eqnarray}
\nonumber
\tilde{p}(a,b,c|x,y,z) & = & 8\sum_{\lambda',\lambda''\in\{0,1\}}\sum_{\lambda\in\Lambda} p(\lambda',\lambda'',\lambda,x,y,z)\delta_{a\oplus\lambda',f_{\lambda}(x)})\delta_{b\oplus\lambda'\oplus\lambda'',f_{\lambda}(y)}\delta_{c\oplus\lambda'',f_{\lambda}(z)}
\end{eqnarray}
in which $p(\lambda',\lambda'') = \frac{1}{4}$ and $p(\lambda',\lambda'',\lambda,x,y,z) = \frac{1}{4}p(\lambda,x,y,z)$ which implies that,
\begin{eqnarray}
\label{eq: nosig}
\nonumber
\sum_{b,c}\tilde{p}(a,b,c|x,y,z)  = \sum_{a,c}\tilde{p}(a,b,c|x,y,z)& = \sum_{a,b}\tilde{p}(a,b,c|x,y,z) = \frac{1}{2}.
\end{eqnarray}

Another property of this distribution concerns the full correlations $\langle\tilde{A}_x\tilde{B}_y\tilde{C}_z\rangle$,
\begin{eqnarray}
\label{eq: full correlator sig -> nosig}
\langle\tilde{A}_x\tilde{B}_y\tilde{C}_z\rangle & = & \langle A_xB_yC_z\rangle.
\end{eqnarray}

From Eq.~\eqref{eq: nosig}, we can see that $\tilde{p}(a,b,c|x,y,z)$ is nosignaling. From Eq.~\eqref{eq: full correlator sig -> nosig} we have that $M(\tilde{p}(a,b,c|x,y,z)) = M^*$, and from the cat that $p(\lambda',\lambda'',\lambda,x,y,z) = p(\lambda',\lambda'')p(\lambda,x,y,z)$ we have that $I(X,Y,Z:\Lambda,\Lambda',\Lambda'') = I(X,Y,Z:\Lambda)$.
\end{proof}

\end{document}